\documentclass[12pt]{article} 
\usepackage[sectionbib]{natbib}
\usepackage{array,epsfig,fancyheadings,rotating}
\usepackage[]{hyperref}  
\usepackage{sectsty, secdot}
\sectionfont{\fontsize{12}{14pt plus.8pt minus .6pt}\selectfont}
\renewcommand{\theequation}{\thesection\arabic{equation}}
\subsectionfont{\fontsize{12}{14pt plus.8pt minus .6pt}\selectfont}

\usepackage{amsmath}
\usepackage{amssymb}
\usepackage{amsfonts}
\usepackage{multirow}
\usepackage{amsthm}

\usepackage{mathtools}
\usepackage{booktabs}
\usepackage{enumitem}
\usepackage{array}
\usepackage{bbm}

\usepackage{times}
\usepackage{bm}
\usepackage{multirow}
\usepackage{graphicx}
\usepackage{setspace}
\usepackage{caption}

\usepackage[titletoc]{appendix}
\usepackage[plain,noend]{algorithm2e}
\usepackage{url} 

\newtheorem{theorem}{Theorem}
\newtheorem{lemma}{Lemma}
\newtheorem{corollary}{Corollary}
\newtheorem{proposition}{Proposition}
\theoremstyle{definition}
\newtheorem{definition}{Definition}
\newtheorem{example}{Example}

\usepackage{color}

\usepackage{mathtools}
\usepackage{booktabs}
\usepackage[english]{babel} 
\usepackage[protrusion=true,expansion=true]{microtype} 
\usepackage{amsmath,amsfonts,amsthm}
\usepackage{ amssymb }
\usepackage{enumitem}
\usepackage{graphicx}
\usepackage{array}

\usepackage{booktabs} 
\usepackage{natbib}
\usepackage{setspace}

\usepackage{microtype} 
\usepackage{tabularx}
\usepackage{subcaption}
\usepackage[font = small,labelfont=bf,textfont=it]{caption} 
\usepackage{footnote}

\usepackage{multirow}

\makeatletter
\def\blfootnote{\gdef\@thefnmark{}\@footnotetext}
\makeatother
\newcolumntype{C}[1]{>{\centering\let\newline\\\arraybackslash\hspace{0pt}}m{#1}}

\DeclareMathOperator*{\argmax}{\arg\!\max}
\renewcommand{\thefootnote}{\fnsymbol{footnote}}
\usepackage{color}
\definecolor{ydcolor}{RGB}{79, 51, 255}

\begin{document}


\renewcommand{\baselinestretch}{2}

\markright{ \hbox{\footnotesize\rm 
}\hfill\\[-13pt]
\hbox{\footnotesize\rm
}\hfill }

\markboth{\hfill{\footnotesize\rm 
Sung et al.} \hfill}
{\hfill {\footnotesize\rm Functional-Input Gaussian Processes} \hfill}

\renewcommand{\thefootnote}{}

$\ $\par


\fontsize{12}{14pt plus.8pt minus .6pt}\selectfont \vspace{0.8pc}
\centerline{\large\bf Functional-Input Gaussian Processes}
\vspace{2pt} \centerline{\large\bf with Applications to Inverse Scattering Problems}
\vspace{.4cm} \centerline{Chih-Li Sung$^{1,*}$, Wenjia Wang\blfootnote{*These authors contributed equally to the manuscript.}$^{2,*}$, Fioralba Cakoni$^{3}$,  Isaac Harris$^{4}$, Ying Hung$^{3}$} \vspace{-0.1cm} {\center{\small\it
$^1$Michigan State University\\ $^2$The Hong Kong University of Science and Technology (Guangzhou)\\$^3$Rutgers, the State University of New Jersey,\quad $^4$Purdue University\\}}
\vspace{.2cm} \fontsize{9}{11.5pt plus.8pt minus
.6pt}\selectfont

\begin{quotation}
\noindent {\it Abstract:}
Surrogate modeling based on Gaussian processes (GPs) has received increasing attention in the analysis of complex problems in science and engineering. Despite extensive studies on GP modeling, the developments for functional inputs are scarce. Motivated by an inverse scattering problem in which functional inputs representing the support and material properties
of the scatterer are involved in the partial differential equations, a new class of kernel functions for functional inputs is introduced for GPs. Based on the proposed GP models, the asymptotic convergence properties of the resulting mean squared prediction errors are derived and the finite sample performance is demonstrated by numerical examples. In the application to inverse scattering, a surrogate model is constructed  with functional inputs, which is crucial to recover the reflective index of an inhomogeneous isotropic scattering region of interest for a given far-field pattern.

\vspace{9pt}
\end{quotation}\par

\def\thefigure{\arabic{figure}}
\def\thetable{\arabic{table}}

\renewcommand{\theequation}{\thesection.\arabic{equation}}
\numberwithin{equation}{section}

\fontsize{12}{14pt plus.8pt minus .6pt}\selectfont

\setcounter{section}{0} 
\setcounter{equation}{0} 

\lhead[\footnotesize\thepage\fancyplain{}\leftmark]{}\rhead[]{\fancyplain{}\rightmark\footnotesize\thepage}

\section{Introduction}

Computer experiments, the studies of real systems using mathematical models such as partial differential equations, have received increasing attention in science and engineering for the analysis of complex problems. Typically, computer experiments require a great deal of time and computing. Therefore, based on a finite sample of computer experiments, it is crucial to build a surrogate for the actual mathematical models and use the surrogate for prediction, inference, and optimization. The Gaussian process (GP) model, also called kriging, is a widely used surrogate model due to its flexibility, interpolating property, and the capability of uncertainty quantification through the predictive distribution. More discussions on computer experiments and surrogate modeling using GP models can be found in  \cite{santner2003design} and  \cite{gramacy2020surrogates}.

This paper is motivated by an inverse scattering problem in computer experiments, where the computer experiments involve functional inputs and therefore the analysis and inference rely on a surrogate model that can take functional inputs into account. 
Figure \ref{fig:inversescattering}  illustrates the idea of inverse scattering. Let the functional input $g$ represent the  material properties of an inhomogeneous isotropic scattering region of interest shown in the middle of Figure \ref{fig:inversescattering}. For a given functional input, the far-field pattern, $u^s$, is obtained by solving partial differential equations \citep{CCHbook} which is computationally intensive.
Given a new far-field pattern, the goal of inverse scattering is to recover the functional input using a surrogate model. Therefore, a crucial step to address this problem is to develop a surrogate model applicable to functional inputs. 
Beyond inverse scattering \citep{CCHbook,kaipio2019bayesian}, problems with functional inputs are frequently found in engineering applications of non-destructive testing, where measurements on the surface or exterior of an object is used to infer the interior structure. 
Similar problems also come up in electrical impedance tomography, where one wishes to recover the functional input representing the electric conductivity, from the measured current to voltage mapping; see, e.g., \cite{dbar2020}, for the electrical impedance tomography model. Another important application is the widely used computerized tomography in medical study for interior reconstruction \citep{courdurier2008solving,li2019learning}.

\begin{figure}[h]
    \centering
    \includegraphics[width=0.47\textwidth]{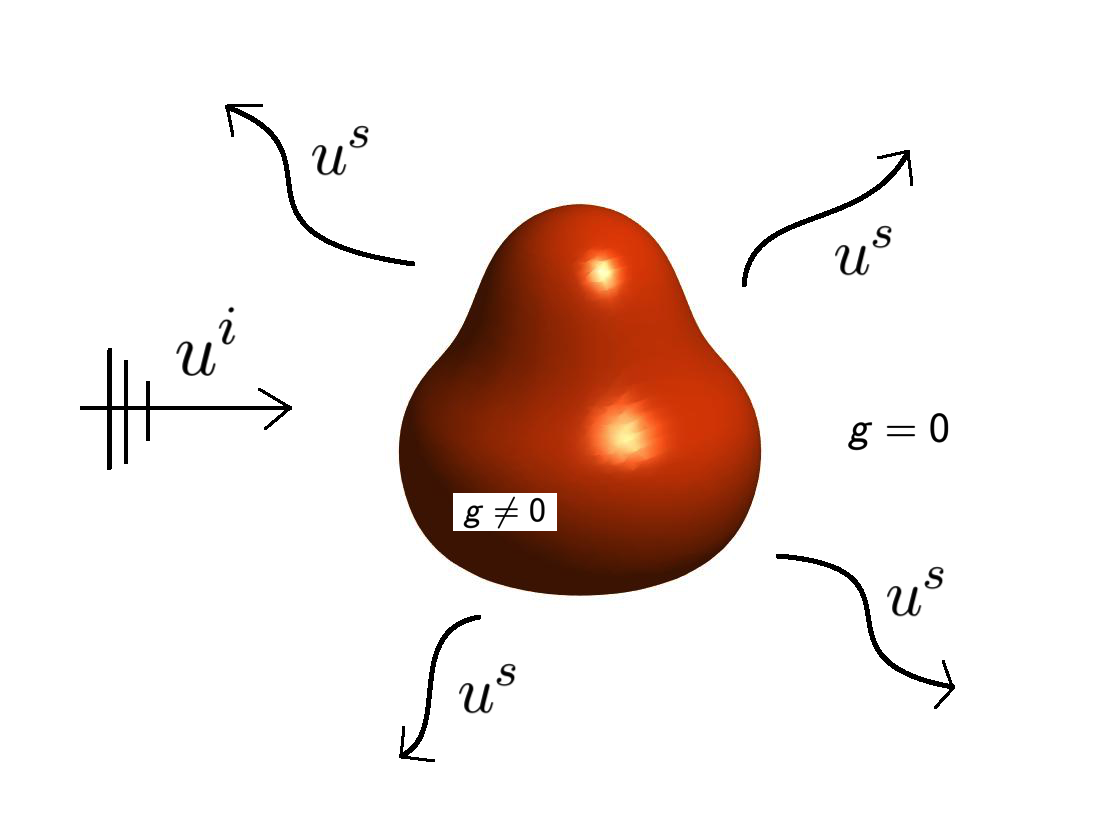}
    \caption{Illustration of the inverse scattering problem.}
    \label{fig:inversescattering}
\end{figure}

Despite extensive studies on surrogate modeling using GPs \citep{gramacy2020surrogates}, the developments for functional inputs are scarce. 
To the best of our knowledge, most of the existing development on GPs involving functional inputs are restrictive to specific applications. 
For example, \cite{nguyen2015gaussian} propose a functional-input GP with bilinear covariance operators and apply it to linear partial differential equations. \cite{morris2012gaussian} develops a kriging model with a covariance function specifically for time-series data. \cite{chen2019function} propose a spectral-distance correlation function and apply it to 3D printing. 

Functional data analysis has been extensively studied in the literature and the research involves functional inputs are often referred to as \textit{scalar-on-function regression} \citep{ramsay2005functional,kokoszka2017introduction, reiss2017methods}. Some approaches reduce the dimension of functional inputs by basis-expansion approximation and then perform a linear or nonlinear model in the reduced Euclidean space (see, e.g., \cite{cardot1999functional,ait2008cross,yao2010functional,muller2013continuously,mclean2014functional}). Another direction is to directly handle the functional inputs by spline approaches (see, e.g., \cite{ferraty2006nonparametric,preda2007regression,baillo2009local,shang2013bayesian}. However, most of these approaches do not incorporate Gaussian process assumptions that allow for uncertainty quantification in the construction of surrogate models.


The focus of this paper is to introduce a new class of Gaussian process (GP) surrogate models for functional inputs.
There are recent studies on surrogate modeling where GP is applied to functional inputs based on truncated basis expansion \citep{shi2008curve,tan2019gaussian,li2021gaussian}. Ideas along this line are intuitive and easy to implement; however, there are three  
drawbacks. First, basis expansion requires an explicit specification of basis functions. Second, basis expansion approximates the functional input and achieves dimension reduction by a finite truncation of the basis functions, which can introduce additional bias to the model. Third, scaling up the techniques developed by basis expansion to high dimensional functional inputs is challenging due to the curse of dimensionality.

To tackle these problems with functional inputs, a new GP surrogate is proposed by introducing a new class of kernel functions that are directly defined on a functional space. It is shown that the proposed kernels are closely connected to the idea of basis expansion without the need of specifying individual basis and without the loss of efficiency due to finite truncation. The procedure is general and provides a parsimonious model especially for high dimensional problems, in which cases basis-expansion approaches often require a great amount of basis functions for high quality approximation. Numerical comparisons with the idea of basis expansion for functional inputs are  conducted in the simulation studies as well as in the application of inverse scattering problem. The 
empirical results of the proposed surrogate model appear to outperform those based on basis expansion in terms of prediction accuracy and uncertainty quantification.

Although the proposed surrogate models extend the conventional GPs to functional inputs, the theoretical results, including the convergence rates of the mean squared prediction errors (MSPE) and the connections to experimental design, are nontrivial extensions. Directly defining the kernels on a functional space reduces the model bias as compared to basis expansion, but posts technical challenges to the theoretical derivations. Additional scattered data approximation techniques, such as the local polynomial reproduction \citep{wendland2004scattered}, have to be rigorously applied to the study of convergence rates. The convergence rates are further explored by the notion of fill distances, which provides a concrete connection between the performance of the proposed model and the experimental design in a functional space.

The remainder of the paper is organized as follows. In Section 2, a functional-input GP model is introduced. A new class of kernel functions, including a linear and a nonlinear kernel, and their theoretical properties are discussed in Section 3. Numerical analysis is conducted in Section 4 to examine the prediction accuracy of the proposed models. In Section 5, the proposed framework is applied to construct a surrogate model for an inverse scattering problem. Concluding
remarks are given in Section 6. Detailed theoretical proofs, and the data and \textsf{R} code for reproducing the numerical results, are provided in Supplementary Materials.


\section{Functional-Input Gaussian Process}
 

Suppose that $V$ is a functional space consisting of functions defined on a compact and convex region $\Omega\subseteq\mathbb{R}^d$, and all functions $g\in V$ are continuous on $\Omega$, i.e., $V\subset C(\Omega)$. A functional-input GP, $f:V\rightarrow \mathbb{R}$, is denoted by
\begin{align}\label{eq:figp}
    f(g)\sim\mathcal{FIGP}(\mu,K(g,g')),
\end{align}
where $\mu$ is an unknown mean and $K(g,g')$ is a semi-positive kernel function for $g,g'\in V$. A new class of kernel $K(g,g')$ for functional inputs is discussed in Section 3.


Given a properly defined kernel function, the estimation and prediction procedures are similar to the conventional GP. 
Assume that there are $n$ realizations from the functional-input GP, where $g_1,\ldots,g_n$ are the inputs and  $f(g_1),\ldots,f(g_n)$ are the outputs. We have $f(g_1),\ldots,f(g_n)$ following a multivariate normal distribution, $\mathcal{N}_n(\boldsymbol{\mu}_n,\mathbf{K}_n)$, with mean $\boldsymbol{\mu}_n=\mu\boldsymbol{1}_n$ and covariance $\mathbf{K}_n$, where $\boldsymbol{1}_n$ is a size-$n$ all-ones vector and $(\mathbf{K}_n)_{j,k}=K(g_j,g_k)$. 
The unknown parameters, including $\mu$ and the hyperparameters  associated with the kernel function, can be estimated by likelihood-based approaches or Bayesian approaches. We refer the details of the estimation methods to \cite{santner2003design} and \cite{gramacy2020surrogates}.

Suppose $g\in V$ is an untried new function. By the property of the conditional multivariate normal distribution, the corresponding output $f(g)$ follows a normal distribution with the mean and variance,
\begin{align}
\mathbb{E}[f(g)|\mathbf{y}_n]= &\mu+\mathbf{k}_n(g)^T \mathbf{K}^{-1}_n (\mathbf{y}_n-\boldsymbol{\mu}_n),\quad \text{and}\label{mean}\\
\mathbb{V}[f(g)|\mathbf{y}_n]= & K(g,g)-\mathbf{k}_n(g)^T \mathbf{K}^{-1}_n \mathbf{k}_n(g),\label{var}
\end{align}
where $\mathbf{y}_n=(y_1,...,y_n)^T$, $y_i=f(g_i)$, and $\mathbf{k}_n(g)=(K(g,g_1),...,K(g,g_n))^T$. The conditional mean of \eqref{mean} is used to predict $f(g)$, and the conditional variance of \eqref{var} can be used to quantify the prediction uncertainty. 


\section{A New Class of Kernel Functions}

In this section, a new class of kernel functions for the functional-input GP is introduced. Based on the proposed models, the asymptotic convergence properties of the resulting mean squared prediction errors are derived. Section \ref{sec:linear} focuses on the discussions of a linear kernel and  Section \ref{sec:nonlinear} extends the discussions to a nonlinear kernel. A practical guidance on the selection of optimal kernel is discussed in Section \ref{sec:selectionkernel}. For notational simplicity, the mean in (\ref{eq:figp}) is assumed to be zero in this section but the results can be easily extended to non-zero cases.

\subsection{Linear kernel for functional inputs}\label{sec:linear}



We first introduce a \textit{linear kernel} for functional inputs $g_1$ and $g_2$:
\begin{align}\label{eq:lk}
    K(g_1,g_2)=\int_{\Omega}\int_{\Omega} g_1(\mathbf{x})g_2(\mathbf{x}')\Psi(\mathbf{x},\mathbf{x}'){\rm{d}}\mathbf{x}{\rm{d}}\mathbf{x}',
\end{align}
where $g_1,g_2\in V$ and $\Psi$ is a positive definite function defined on $\Omega\times \Omega$. It can be shown in the following proposition that this kernel function is semi-positive definite.

\begin{proposition}\label{Prop1}
The linear kernel $K$ defined in (\ref{eq:lk}) is semi-positive definite on $V\times V$.
\end{proposition}

By Mercer's theorem \citep{rahman2007integral}, we have
\begin{align}\label{eq:mercer1}
\Psi(\mathbf{x},\mathbf{x}') =\sum_{j=1}^\infty \lambda_j \phi_j(\mathbf{x})\phi_j(\mathbf{x}'),
\end{align}
where $\mathbf{x},\mathbf{x}'\in \Omega$, and
$\lambda_1\geq \lambda_2\geq...> 0$ and 
$\{\phi_k\}_{k\in\mathbb{N}}$ are the eigenvalues and the orthonormal basis in $L_2(\Omega)$, respectively. Given the positive definite function $\Psi$, we can construct a GP via the  Karhunen--Lo\`eve expansion:
\begin{align}\label{eq:lifigp}
    f(g) = \sum_{j=1}^\infty \sqrt{\lambda_j}\langle \phi_j,g \rangle_{L_2(\Omega)}Z_j,
\end{align}
where $Z_j$'s are independent standard normal random variables, and $\langle\phi_j,g \rangle_{L_2(\Omega)}$ is the inner product of $\phi_j$ and $g$, which is $\langle\phi_j,g \rangle_{L_2(\Omega)}=\int_\Omega \phi_j(\mathbf{x})g(\mathbf{x}){\rm{d}}\mathbf{x}$.
It can be shown that the covariance function of the constructed GP in (\ref{eq:lifigp}) is  $K(g_1,g_2)$ defined in (\ref{eq:lk}), i.e., 
\begin{align}\label{eq:covli}
    {\rm Cov}(f(g_1),f(g_2)) = & \sum_{j=1}^\infty \lambda_j\langle \phi_j,g_1 \rangle_{L_2(\Omega)}\langle \phi_j,g_2 \rangle_{L_2(\Omega)}\nonumber\\
    =&  \sum_{j=1}^\infty \lambda_j \int_\Omega\int_\Omega g_1(\mathbf{x})\phi_j(\mathbf{x})g_2(\mathbf{x}')\phi_j(\mathbf{x}'){\rm{d}}\mathbf{x}{\rm{d}}\mathbf{x}'\nonumber\\
    =&  \int_\Omega\int_\Omega g_1(\mathbf{x})g_2(\mathbf{x}')\Psi(\mathbf{x},\mathbf{x}'){\rm{d}}\mathbf{x}{\rm{d}}\mathbf{x}'
\end{align}
for any $g_1,g_2\in V$. 

The proposed surrogate model is equivalent to a basis expansion through the Karhunen--Lo\`eve expansion in \eqref{eq:lifigp}, but it is worth noting that the proposed method only requires  the specification of kernel function in  \eqref{eq:lk} instead of an explicit specification of individual basis $\phi_j$. Furthermore, there is no dimension reduction or approximation applied to the functional input, and thus there is no additional bias introduced to the surrogate.  
More specifically, the proposed model preserves the most information  without finite truncation of basis expansion because the kernel representation \eqref{eq:lk} is  equivalent to representing each input $g$ as an element in $L_2(\Omega)$ through a basis expansion with respect to $\{\phi_j\}^{\infty}_{j=1}$. 
These advantages as compared to basis expansion are commonly seen in kernel-based methods, such as support-vector machines (SVM), kernel principal components analysis (KPCA), and kernel ridge regression (KRR) \citep{friedman2017elements}. 

\begin{proposition}\label{Proplinear}
The Gaussian process, $f(g)$, constructed as in \eqref{eq:lifigp} is linear, i.e., for any $a,b\in \mathbb{R}$ and $g_1,g_2\in V$, it follows that $f(ag_1+bg_2)=af(g_1)+bf(g_2)$.
\end{proposition}



The proposed kernel function has an intuitive interpretation that connects to Bayesian modeling. In a Bayesian linear regression, the conditional mean function is assumed to be $f(\mathbf{x})=\mathbf{x}^T\mathbf{w}$, where $\mathbf{w}$ is typically assumed to have a multivariate normal prior, i.e., $\mathbf{w}\sim \mathcal{N}(0, \Sigma_\mathbf{w})$. Hence, for any two points, $\mathbf{x}$ and $\mathbf{x}'$, the covariance of $f(\mathbf{x})$ and $f(\mathbf{x}')$ is  $\text{Cov}(f(\mathbf{x}'),f(\mathbf{x}'))=\mathbf{x}^T\Sigma_\mathbf{w} \mathbf{x}'$, which can be interpreted as a \textit{weighted} inner product of $\mathbf{x}$ and $\mathbf{x}'$. The proposed model \eqref{eq:lifigp}  can be viewed as an analogy to the Bayesian linear model with functional inputs and the covariance \eqref{eq:covli} can be viewed as a weighted inner product of the two functions $g_1$ and $g_2$.


To understand the performance of the proposed predictor of \eqref{mean} with the kernel function defined in \eqref{eq:lk}, we first characterize the mean squared prediction error in the following theorem. Denote the reproducing kernel Hilbert space (RKHS) associated with a kernel $\Psi$ as $\mathcal{N}_{\Psi}(\Omega)$.

\begin{theorem}\label{thm:linear}
Let $\hat{f}(g)=\mathbb{E}[f(g)|\mathbf{y}_n]$ as in \eqref{mean}. For any continuous function $g\in V\subset L_2(\Omega)$, define a linear operator on $L_2(\Omega)$:
\begin{align*}
    \mathcal{T}g(\mathbf{x}) = \int_\Omega g(\mathbf{x}')\Psi(\mathbf{x},\mathbf{x}'){\rm{d}}\mathbf{x}'.
\end{align*}
The mean squared prediction error (MSPE) of $\hat{f}(g)$ can be written as 
\begin{align}\label{eq:thmlieq}
    \mathbb{E}\left(f(g)-\hat f(g)\right)^2
    = \min_{(u_1,...,u_n)\in \mathbb{R}^n} \left\|\mathcal{T}g - \sum_{j=1}^n u_j \mathcal{T}g_j\right\|_{\mathcal{N}_\Psi(\Omega)}^2,
\end{align}
where $\left\|\cdot\right\|_{\mathcal{N}_\Psi(\Omega)}$ is the RKHS norm of $\mathcal{N}_\Psi(\Omega)$.
\end{theorem}

By Proposition 10.28 in \cite{wendland2004scattered}, it is can be shown that $\mathcal{T}g\in \mathcal{N}_\Psi(\Omega)$ and therefore the right-hand side of \eqref{eq:thmlieq} exists. The theorem provides a new representation of the MSPE for functional-input GPs that is analogue to that of conventional GPs in the $L_2$ input space, which has not been yet explored in the existing literature.
According to Theorem \ref{thm:linear}, the MSPE can be represented as the distance between $\mathcal{T}g$ and its projection on the linear space spanned by $\{\mathcal{T}g_1,...,\mathcal{T}g_n\}$. This distance can be reduced if $g_j$'s can be designed so that  the spanned space can well approximate the space $V$.
We highlight some designs of $g_j$'s in the following two corollaries where the convergence rates of MSPE can be explicitly discussed.

In the following corollaries, the kernel function $\Psi$ is assumed to be a Mat\'ern kernel \citep{stein2012interpolation}: 
\begin{align}
\Psi(\mathbf{x},\mathbf{x}')=\psi(\|\Theta(\mathbf{x}-\mathbf{x}')\|_2)\quad\text{with}\label{matern2}\\
\psi(r)=
\frac{\sigma^2}{\Gamma(\nu)2^{\nu-1}}(2\sqrt{\nu} r)^\nu B_\nu(2\sqrt{\nu}r),\label{matern1}
\end{align}
where $\Theta$ is a lengthscale parameter, which is a $d\times d$ positive diagonal matrix, $\|\cdot\|_2$ denotes the Euclidean norm, $\sigma^2$ is a positive scalar, $B_\nu$ is the modified Bessel function of the second kind, and $\nu$ represents a smoothness parameter.
The Mat\'ern kernel is considered here because it has been widely used in the computer experiments and spatial statistics literature \citep{santner2013design,stein2012interpolation}. The corollaries can be also extended to a general positive kernel which has $k$ continuous derivatives, such as Wendland’s compactly supported kernel \citep{wendland2004scattered}. We refer the details of this extension to \cite{wendland2004scattered} and \cite{haaland2011accurate}.
Without loss of generality, we assume that $\Theta$ is an identity matrix and $\sigma^2=1$ for the theoretical developments in this section. More detailed discussions of these parameters, including $\Theta,\sigma^2$ and $\nu$, are given in Section \ref{sec:numericstudy}.
 

\begin{corollary}\label{coro:linear1} Suppose $g_j$, $j=1,\ldots,n$ are the first $n$ eigenfunctions of $\Psi$, i.e, $g_j=\phi_j$. 
For $g\in V\subset L_2(\Omega)$, there exists a constant $ C_1>0$ such that
\begin{align}\label{eq:l1c11}
    \mathbb{E}\left(f(g)-\hat f(g)\right)^2\leq C_1 \|g\|_{L_2(\Omega)}^2n^{-\frac{2\nu}{d}}.
\end{align}
Furthermore, if $g\in \mathcal{N}_{\Psi}(\Omega)$, then there exists a constant $C_2>0$ such that
\begin{align}\label{eq:l1c12}
    \mathbb{E}\left(f(g)-\hat f(g)\right)^2\leq C_2 \|g\|_{\mathcal{N}_{\Psi}(\Omega)}^2n^{-\frac{4\nu}{d}}.
\end{align}
\end{corollary}

Corollary \ref{coro:linear1} represents the convergence rate analogue to that of conventional GPs \citep{tuo2020kriging}, and shows that if we can design the input functions to be eigenfunctions of $\Psi$, the convergence rate of MSPE is polynomial. If the functional space is further assumed to be the RKHS associated with the kernel $\Psi$ (i.e. $g\in \mathcal{N}_{\Psi}(\Omega)$), which is
smaller than $L_2(\Omega)$, the convergence rate becomes faster as indicated by (\ref{eq:l1c12}). 
This result indicates a significant difference between the proposed GP defined on a functional space and the conventional one defined on a Euclidean space. 
That is, the convergence results of (\ref{eq:l1c11}) and (\ref{eq:l1c12}) depend on the norm of the functional space that the input $g$ lies in, which is different from that of conventional GPs which only involves the Euclidean norm.

Instead of selecting the input functions to be eigenfunctions, an alternative is to design the input functions through a set of \textit{knots} in $\Omega$, i.e.,  $\mathbf{X}_n\equiv\{\mathbf{x}_1,\ldots,\mathbf{x}_n\}$, where $\mathbf{x}_j \in \Omega$ for $j=1,\ldots,n$, and the convergence rate is derived in the following corollary.  We first denote $h_{\mathbf{X}_n,\Omega}$ as the \textit{fill distance} of $\mathbf{X}_n$, i.e., \begin{align*}
h_{\mathbf{X}_n,\Omega}:= \sup_{\mathbf{x}\in\Omega}\min_{\mathbf{x}_j\in \mathbf{X}_n}\|\mathbf{x}-\mathbf{x}_j\|_2.
\end{align*} 
Further, denote $q_{\mathbf{X}_n} = \min_{1\leq j\neq k\leq n}\|\mathbf{x}_j-\mathbf{x}_k\|/2$, and a design $\mathbf{X}_n$ satisfying $h_{\mathbf{X}_n,\Omega}/q_{\mathbf{X}_n}\leq C'$ for some constant $C'$ is called a \textit{
quasi-uniform design}.

\begin{corollary}\label{coro:linear2}
\begin{itemize}
\item[(1)] Suppose $g_j(\mathbf{x})=\Psi(\mathbf{x},\mathbf{x}_j)$, where $\mathbf{x},\mathbf{x}_j\in \Omega$ for $j=1,\ldots,n$. For $g\in \mathcal{N}_{\Psi}(\Omega)$, there exists a constant $C_3>0$ such that
\begin{align*}
    \mathbb{E}\left(f(g)-\hat f(g)\right)^2\leq C_3\|g\|_{\mathcal{N}_{\Psi}(\Omega)}^2 h_{\mathbf{X}_n,\Omega}^{2\nu}.
\end{align*}
\item[(2)]  For a quasi-uniform design $\mathbf{X}_n$, there exists a positive constant $C$ such that $h_{\mathbf{X}_n,\Omega} \leq  C n^{-1/d}$ \citep{wendland2004scattered,muller2009komplexitat}. Therefore,  there exists a constant $C_4>0$ such that
\begin{align}\label{corollary35}
    \mathbb{E}\left(f(g)-\hat f(g)\right)^2\leq C_4\|g\|_{\mathcal{N}_{\Psi}(\Omega)}^2  n^{-\frac{2\nu}{d}}.
\end{align}
\end{itemize}
\end{corollary}

As compared with the results in Corollary \ref{coro:linear1},
the convergence rate of quasi-uniform designs as shown in (\ref{corollary35}) is slower than the choice of eigenfunctions as shown in (\ref{eq:l1c12}). Despite a slower rate of convergence, designing functional inputs by 
$\Psi(\mathbf{x}_j,\cdot)$ with space-filling $\mathbf{x}_j$'s can be relatively easier in practice than finding eigenfunctions of $\Psi$. However, if the eigenfunctions of $\Psi$ are available, then the design based on Corollary \ref{coro:linear1} (i.e., the first $n$ eigenfunctions) would be recommended as the convergence rate is faster. Some kernel functions exist  closed form
expressions, such as Gaussian kernels \citep{zhu1997gaussian}. More generally, the eigenfunctions can be numerically approximated  via Nystr{\"o}m's method \citep{williams2000using}.



The proposed linear kernel can be naturally modified to accommodate the potential non-linearity in $f$ by enlarging the feature space using a pre-specified nonlinear transformation $\mathcal{M}$ on $g$, i.e., $\mathcal{M} :V\rightarrow V_1$, where $V_1$ is a function class. The resulting  kernel function can be written as 
\begin{align*}
    K(g_1,g_2)=\int_\Omega\int_\Omega \mathcal{M}\circ g_1(\mathbf{x})\mathcal{M}\circ g_2(\mathbf{x}')\Psi(\mathbf{x},\mathbf{x}'){\rm{d}}\mathbf{x}{\rm{d}}\mathbf{x}',
\end{align*}
and the corresponding GP can be constructed by
\begin{align}\label{eq:nlgpfromlgp}
    f(g) = \sum_{j=1}^\infty \sqrt{\lambda_j}\langle \phi_j,\mathcal{M}\circ g \rangle_{L_2(\Omega)}Z_j.
\end{align}
The convergence results of MSPE can be extended to \eqref{eq:nlgpfromlgp}. 
There are many possible ways to define $\mathcal{M}$ so that the feature space can be enlarged; however, the flexibility comes with a higher estimation complexity. In the next section, we propose an alternative to address the non-linearity through a kernel function, which is computationally more efficient.

\subsection{Nonlinear kernel for functional inputs}\label{sec:nonlinear}

In this section, we introduce a new type of kernel function for functional inputs that takes into account the non-linearity via a radial basis function. 
Let $\psi(r):\mathbb{R}^{+}\rightarrow\mathbb{R}$ be a radial basis function whose corresponding kernel in $\mathbb{R}^d$ is strictly positive definite for any $d\geq 1$. 
Note that the radial basis function of \eqref{matern1}, whose corresponding kernel is a Mat\'ern kernel, satisfies this condition.
Define $K:V\times V\rightarrow \mathbb{R}$ as
\begin{align}\label{eq:nonlinearkernel}
    K(g_1,g_2) = \psi(\gamma\|g_1-g_2\|_{L_2(\Omega)}),
\end{align}
where $\|\cdot\|_{L_2(\Omega)}$ is the $L_2$-norm of a function, defined by $\|g\|_{L_2(\Omega)}=(\langle g,g\rangle_{L_2(\Omega)})^{1/2}$, and $\gamma>0$ is a parameter that controls the decay of the kernel function with respect to the $L_2$-norm. 


While it is of great interest to consider other distance metrics to define the distance between two functions, such as Fr\'echet distance and $L_{\infty}$-norm, such resulting kernel functions is not necessary to be semi-positive definite, which is a required property for defining a kernel function. For example, consider an $L_{\infty}$-norm distance for the kernel function, that is,
$K(g_1,g_2) = \psi(\gamma\|g_1-g_2\|_{L_{\infty}(\Omega)}),$
for any $g_1,g_2\in L_{\infty}(\Omega)$, where $\psi$ has the form of \eqref{matern1} with $\nu=2.5$ and $\sigma^2=1$. Given the four training functional inputs, $g_1(x_1,x_2)=x_1^2,g_2(x_1,x_2)=x_2^2,g_3(x_1,x_2)=1+x_1,g_4(x_1,x_2)=1+x_2$, and $\gamma=0.5$, the kernel matrix is 
$$\mathbf{K}_n=
\left[\begin{array}{cccc}
     1& 0.8286& 0.7536& 0.5240\\
     0.8286& 1& 0.5240 &0.7536\\
     0.7536& 0.5240& 1 &0.8286\\
     0.5240& 0.7536& 0.8286 &1
\end{array}\right].
$$
Then, for a vector $\mathbf{a}=(1,-1,-1,1)^T$, it follows that $\mathbf{a}^T\mathbf{K}_n\mathbf{a}=-0.2331<0$, which implies that the kernel function is not semi-positive. Conditions on the metric $\|\cdot\|$ such that the resulting kernel function is positive definite will be pursued in the future. In the following proposition, we show that the kernel function with $\|\cdot\|_{L_2}$, defined as in \eqref{eq:nonlinearkernel}, is positive definite.

\begin{proposition}\label{prop:pdofnl}
The function $K$ defined in (\ref{eq:nonlinearkernel}) is positive definite on $V\times V$.
\end{proposition}

Assume that there exists a probability measure $P$ on $V$ such that $\int_V g(t)^2 dP(g) <\infty$, for $t\in \Omega$ \citep{ritter2007average}.
Based on the positive definite function $K$ as in \eqref{eq:nonlinearkernel}, we can construct a GP via the Karhunen--Lo\`eve expansion:
\begin{align}\label{eq:nlexp}
    f(g)=\sum_{j=1}^\infty \sqrt{\lambda_j}\varphi_j(g)Z_j,
\end{align} 
where $\varphi_j$'s are the orthonormal basis obtained by a generalized version of Mercer's theorem, $ K(g_1,g_2)=\sum_{j=1}^\infty \lambda_j\varphi_j(g_1)\varphi_j(g_2)$ \citep{steinwart2012mercer} with respect to the probability measure $P$.

The nonlinear kernel in \eqref{eq:nonlinearkernel} can be viewed as a basis expansion of the functional input based on the fact that 
$\|g_1-g_2\|^2_{L_2(\Omega)} = \sum^\infty_{j= 1}\langle\phi_j,g_1-g_2\rangle_{L_2(\Omega)}^2$, where $\{\phi_j\}^\infty_{j=1}$ are orthonormal basis functions in $L_2(\Omega)$. By a finite truncation of basis expansion, the input $g$ can be approximated by  $\{\phi_j\}^M_{j=1}$ in $\mathbb{R}^M$, for a positive integer $M$, and therefore $f(g)$ can be approximated by a GP with the correlation function, $\psi(\gamma\|\cdot\|_2)$, where $\|\cdot\|_2$ is the Euclidean norm on $\mathbb{R}^M$. However, similar to the discussions in Section \ref{sec:linear}, this introduces additional bias due to the finite truncation and requires an explicit specification of the orthonormal functions $\{\phi_j\}^M_{j=1}$ and $M$ in advance. Instead, the proposed method directly evaluates the correlation on a functional space through the nonlinear kernel without approximation, and the application only requires  the selection of a proper kernel function.

It is also worth noting that the $L_2$ norm in \eqref{eq:nonlinearkernel} can be replaced by any Hilbert space norm, such as RKHS norm. Therefore,
the nonlinear kernel \eqref{eq:nonlinearkernel} is flexible and can be generalized to any target space of interest in practice.
Nevertheless, the $L_2$ norm  can be approximated by numerical integration methods, such as Monte Carlo integration \citep{caflisch1998monte}, which is computationally more efficient compared to, for example, the RKHS norm that requires the computation of inverting an $N\times N$ matrix, where $N$ is the size of grid points. 

Based on the proposed nonlinear kernel, the convergence rate of the MSPE is studied in the following theorem.



\begin{theorem}\label{coro:nlmatern}
Suppose that $\Phi$ is a Mat\'ern kernel function with smoothness $\nu_1$, and $\psi$ is the radial basis function of \eqref{matern1}, whose corresponding kernel is Mat\'ern with smoothness $\nu$.
Let $\tau = \min(\nu,1)$. 
For any $n>N_0$ with  a constant $N_0$, there exist $n$ input functions such that for any $g\in \mathcal{N}_{\Phi}(\Omega)$ with $\|g\|_{\mathcal{N}_{\Phi}(\Omega)}\leq 1$, the MSPE can be bounded by
\begin{align}\label{eq:nlmaternX}
    \mathbb{E}\left(f(g)-\hat f(g)\right)^2 \leq C_5(\log n)^{-\frac{(\nu_1+d/2)\tau}{d}}\log \log n.
\end{align}
\end{theorem}

Based on (\ref{eq:nlmaternX}), it appears that the convergence rate is slower than the conventional GP where the inputs are defined in the Euclidean space. Although potentially there are rooms to further sharpen this rate, a slower rate of convergence for functional inputs is not surprising because the input space is much larger than the Euclidean space. Note that since the RKHS generated by a Mat\'ern kernel function with smoothness $\nu_1$ is equivalent to the (fractional) Sobolev space $H^{\nu_1+d/2}(\Omega)$ \citep{wendland2004scattered}, the assumption of $g\in \mathcal{N}_{\Phi}(\Omega)$ in Theorem  \ref{coro:nlmatern} is equivalent to $g\in H^{\nu_1+d/2}(\Omega)$.

If $\Phi$ is a squared exponential kernel, the corresponding RKHS is within the RKHS generated by a Mat\'ern kernel function with any smoothness $\nu_2$. Thus, one can choose a large $\nu_2>\nu_1$ and apply Theorem \ref{coro:nlmatern} to obtain the same convergence rate as in (\ref{eq:nlmaternX}) by replacing $\nu_1$ with $\nu_2$. Therefore, one can conclude that the convergence rate of RKHS generated by a squared exponential kernel is faster than that of the RKHS generated by a Mat\'ern kernel function with a fixed $\nu_1$.

\subsection{Selection of kernels}\label{sec:selectionkernel}
Based on Sections 3.1 and 3.2, the linear kernel of \eqref{eq:lk} results in a less flexible model leading to a lower prediction variance but higher bias, while the nonlinear kernel of \eqref{eq:nonlinearkernel} results in a more flexible model leading to a higher variance but lower bias \citep{friedman2017elements}. 
To find an optimal kernel function that balances the bias--variance trade-off, the idea of cross-validation is adopted that allows us to select the kernel by minimizing the estimated prediction error.

Although cross-validation methods are typically expensive to implement in many situations, the leave-one-out cross-validation (LOOCV) of GP models can be expressed in a closed form, which makes computation less demanding \citep{zhang2010kriging,rasmussen2006gaussian,currin1988bayesian}. 
Specifically, denote $\tilde{y}_i$ as the prediction mean based on all data except $i$th observation and $y_i$ as the real output of $i$th observation. One can show that, for a kernel candidate $K$, which can be either the linear kernel \eqref{eq:lk} or the nonlinear kernel \eqref{eq:nonlinearkernel},  the LOOCV error is
\begin{equation}\label{eq:LOOCV}
    \frac{1}{n}\sum^{n}_{i=1}(y_i-\tilde{y}_i)^2=\frac{1}{n}\| \boldsymbol{\Lambda}_n^{-1}\mathbf{K}_n^{-1}(\mathbf{y}_n -\boldsymbol{\mu}_n)\|^2_2,
\end{equation}
where 
$\boldsymbol{\Lambda}_n$ is a diagonal matrix with the element $(\boldsymbol{\Lambda}_n)_{j,j}=(\mathbf{K}_n^{-1})_{j,j}$. Thus, between the linear and nonlinear kernels, the optimal one can be  selected by minimizing the LOOCV error.  


\subsection{Generalization to multiple functional-input variables}
Both of the linear and nonlinear kernel functions developed in Sections \ref{sec:linear} and \ref{sec:nonlinear} can be naturally extended to multiple functional-input variables. For example, suppose that there are two functional-input variables, $g\in V$ and $h\in V$, and the $n$ inputs, $\{(g_1,h_1), \ldots,(g_n,h_n)\}$, are collected. In such cases, the linear kernel \eqref{eq:lk} can be rewritten as
\begin{align*}
    K((g_1,h_1),(g_2,h_2))=\int_{\Omega}\int_{\Omega} \left(g_1(\mathbf{x})g_2(\mathbf{x}')+h_1(\mathbf{x})h_2(\mathbf{x}')\right)\Psi(\mathbf{x},\mathbf{x}'){\rm{d}}\mathbf{x}{\rm{d}}\mathbf{x}',
\end{align*}
and the nonlinear kernel \eqref{eq:nonlinearkernel} can be rewritten as
\begin{align*}
    K((g_1,h_1),(g_2,h_2)) = \psi\left(\left(\gamma_1\|g_1-g_2\|_{L_2(\Omega)}^2+\gamma_2\|h_1-h_2\|_{L_2(\Omega)}^2\right)^{1/2}\right),
\end{align*}
where $\gamma_1,\gamma_2>0$ are the parameters.

The nonlinear kernel also can be naturally generalized to the mixture of functional inputs and scalar inputs. That is, suppose that in addition to the two functional-input variables, $g,h\in V$, there exists a scalar input variable in the experiment, denoted by $z\in\Omega'\subseteq\mathbb{R}$, then a kernel function can be defined as 
\begin{align*}
    K((g_1,h_1,z_1)&,(g_2,h_2,z_2)) =\\
    &\psi\left(\left(\gamma_1\|g_1-g_2\|_{L_2(\Omega)}^2+\gamma_2\|h_1-h_2\|_{L_2(\Omega)}^2+\gamma_3(z_1-z_2)^2\right)^{1/2}\right),
\end{align*}
where $\gamma_3>0$.

\section{Numerical Study}\label{sec:numericstudy}

In this section, numerical experiments are conducted to examine the emulation performance of the proposed method. In Supplementary Material \ref{sec:samplepath}, the sample paths of the functional-input GP with different parameter settings are explored.

In these numerical studies, the quasi-Monte Carlo integration \citep{morokoff1995quasi} is used to numerically evaluate the integrals in the kernels. Specifically, suppose that $\Omega$ is a unit cube, then the linear kernel \eqref{eq:lk} can be approximated by
\begin{equation}\label{eq:linearkernelMC}
K(g_1,g_2)\approx\frac{1}{N^2}\sum^N_{i=1}\sum^N_{j=1} g_1(\mathbf{x}_i)g_2(\mathbf{x}'_j)\Psi(\mathbf{x}_i,\mathbf{x}'_j),
\end{equation}
where $\{\mathbf{x}_i\}^N_{i=1}$ and $\{\mathbf{x}'_j\}^N_{j=1}$ are low-discrepancy sequences from a unit cube, for which the Sobol sequence \citep{sobol1967distribution,bratley1988algorithm} is considered here. The number of points in the sequence, $N=5,000$, is set. 
Similarly, the $L_2$-norm in the nonlinear kernel \eqref{eq:nonlinearkernel} can be approximated by
\begin{equation}\label{eq:nonlinearkernelMC}
\|g_1-g_2\|_{L_2(\Omega)}\approx\left(\frac{1}{N}\sum^N_{i=1}(g_1(\mathbf{x}_i)-g_2(\mathbf{x}_i))^2\right)^{1/2}.
\end{equation}

The prediction performance of the proposed method is examined by three synthetic examples, including a linear operator, $f_1(g)=\int_\Omega\int_\Omega g(\mathbf{x}) {\rm{d}}x_1{\rm{d}}x_2$,
and two nonlinear ones, $f_2(g)=\int_\Omega\int_\Omega g(\mathbf{x})^3  {\rm{d}}x_1{\rm{d}}x_2$ and $f_3(g)=\int_\Omega\int_\Omega\sin(g(\mathbf{x})^2)  {\rm{d}}x_1{\rm{d}}x_2$, where $\mathbf{x}=(x_1,x_2)\in \Omega\equiv[0,1]^2$ and $g(\mathbf{x}):[0,1]^2\rightarrow \mathbb{R}$.
Eight functional inputs, which are shown in the first row of Table \ref{tab:linear_simulation}, are considered for each of the synthetic examples and their outputs are given in Table \ref{tab:linear_simulation}.

Three types of functional inputs are tested for predictions: $g_9(\mathbf{x})=\sin(\alpha_1x_1+\alpha_2 x_2),g_{10}(\mathbf{x})=\beta+x_1^2+x_2^3$, and $g_{11}(\mathbf{x})=\exp\{-\kappa x_1x_2\}$, where $\alpha_1,\alpha_2,\beta,\kappa\in[0,1]$. 
Based on 100 random samples of $\alpha_1,\alpha_2,\beta$ and $\kappa$ from $[0,1]$, the prediction performance is evaluated by averaging the mean squared errors (MSEs), where
$
\text{MSE}=\frac{1}{3}\sum^{11}_{j=9}(f(g_j)-\hat{f}(g_j))^2.
$


The proposed method is performed with
the Mat{\'e}rn kernel function with the smoothness parameter $\nu=5/2$, which leads to a simplified form of \eqref{matern1}:
\begin{equation}\label{eq:maternkernel2.5}
    \psi(r)=\left(1+\sqrt{5}r+\frac{5}{3}r^2\right)\exp\left(-\sqrt{5}r\right).
\end{equation}
Other parameters, including $\Theta,\sigma^2$ and $\gamma$, are estimated via maximum likelihood estimation. 
Both of the linear kernel \eqref{eq:lk} and the nonlinear kernel \eqref{eq:nonlinearkernel} are performed for the proposed functional input GP and their LOOCV errors are reported in Table \ref{tab:prediction_simulation}. According to Section 3.3, LOOCV is then used to identify the optimal kernel. By minimizing the LOOCV errors, the linear kernel is identified as the optimal choice for the linear synthetic example, $f_1(g)$, and the nonlinear kernel is identified as the optimal choice for the nonlinear synthetic examples, $f_2(g)$ and $f_3(g)$. For the three synthetic examples, their MSEs are summarized in Table \ref{tab:prediction_simulation}. It appears that the optimal kernels selected by LOOCV are consistent with the selections based on minimizing the MSEs, which shows that LOOCV is a reasonable indicator of the optimal kernel when the ground truth is unknown. 

\begin{table}[t]
\begin{center}
\begin{tabular}{ c|c|c|c|c } 
 \toprule
Measurements & Kernel  & $f_1(g)=\int_\Omega\int_\Omega g$ & $f_2(g)=\int_\Omega\int_\Omega g^3$& $f_3(g)=\int_\Omega\int_\Omega \sin(g^2)$ \\ 
 \midrule
 \multirow{2}{*}{LOOCV} &  linear& $\boldsymbol{7.9\times 10^{-7}}$ & 1.813 & 0.454\\ 
 & nonlinear& $2.2\times 10^{-6}$ & \textbf{0.227} & \textbf{0.017}\\ 
 \midrule
\multirow{2}{*}{MSE} &  linear& $\boldsymbol{6.4\times 10^{-10}}$ & 1.087& 0.140\\ 
 & nonlinear& $3.1\times 10^{-7}$ & \textbf{0.012} & \textbf{0.016}\\
 \bottomrule
\end{tabular}
\end{center}
    \caption{The leave-one-out cross-validation errors (LOOCVs) and the mean squared errors (MSEs)  for three testing functions, in which the errors corresponding to the optimal kernel are boldfaced.}
    \label{tab:prediction_simulation}
\end{table}

The computational cost is also assessed for the two kernel choices. The numerical experiments were performed on a MacBook Pro laptop with Apple M1 Max of Chip and 32 GB of RAM. The computation for linear kernels in each of the examples takes about 9 seconds, while that for nonlinear kernels takes less than 1 second, indicating that the linear kernel requires more computation than the nonlinear kernel. This is not surprising, because the computation for linear kernels involves double integrals (see \eqref{eq:lk}) which in turn requires $N^2$ evaluations for the quasi-Monte Carlo integration as in  \eqref{eq:linearkernelMC}, while the nonlinear kernel (see \eqref{eq:nonlinearkernel} and \eqref{eq:nonlinearkernelMC}) only requires $N$ evaluations. Furthermore, the linear kernel has $d$ lengthscale parameters that need to be estimated, while the nonlinear kernel only has one lengthscale parameter. Nonetheless, fitting the functional-input GP model is reasonably efficient with either a linear or nonlinear kernel, both of which take less than 10 seconds.

As a comparison, we consider a \textit{basis-expansion} approach discussed in Section \ref{sec:nonlinear}. That is, consider a functional principal component analysis (\texttt{FPCA}) with \textit{truncated} components \citep{rice1991estimating,wang2016functional}:
$$
g_i(\mathbf{x})\approx u(\mathbf{x})+\sum^M_{j=1}z_{ij}\psi_j(\mathbf{x}),
$$
with the leading $M$ eigenfunctions $\{\psi_j(\mathbf{x})\}^M_{j=1}$ and the corresponding coefficients $\{z_{ij}\}$ given by:
\begin{align}\label{eq:klexpansion}
\psi_j(\mathbf{x}) &= \argmax_{\substack{\| \phi \|_2 = 1, \\ \langle \phi, \psi_l \rangle = 0, \forall l < j}} \sum^n_{i=1} \left\{ \int \left(g_i(\mathbf{x})-u(\mathbf{x})\right) \phi(\mathbf{x})  d\mathbf{x}\right\}^2, \nonumber\\
z_{ij} &= \int \left(g_i(\mathbf{x})-u(\mathbf{x})\right) \psi_j(\mathbf{x}) \; d\mathbf{x},
\end{align}
where $n=8$ in this example. The number of components, $M=3$, is chosen that explains 99.46\% of variance. We refer more details regarding the expansion to \cite{wang2016functional} and \cite{mak2017efficient}. Given the training input-output pairs, $\{\mathbf{z}_i, y_i\}^{n}_{i=1}$ where $\mathbf{z}_i=(z_{i,1},\ldots,z_{i,M})$, a conventional GP (with a Mat{\'e}rn kernel) is adopted to fit the training data.
The test input, $\{\mathbf{z}_i\}^{11}_{i=9}$ can be similarly obtained by \eqref{eq:klexpansion},  and their outputs are predicted by the fitted GP.

In addition to FPCA, we also consider a Maclaurin series expansion of degree 3, which is a Taylor series expansion of a function evaluated at 0 truncated to degree 3 (labeled \texttt{T3}). That is, 
$$
g_j(\mathbf{x})\approx \sum_{\substack{a=0,b=0\\ a+b\leq 3}}\frac{\partial^{a+b} g_j(0,0)}{\partial x^{a}_1\partial x^b_2}x_1^{a}x_2^{b}
$$
The series expansion can  approximate the functional inputs of the examples reasonably well with only a few non-zero coefficients. For example, the training functional input $g_1(\mathbf{x})=x_1+x_2$, has the coefficient 1 for both terms of $x_1$ and $x_2$ and 0 for other terms.

To evaluate the prediction performance and quantify the uncertainty, in addition to MSEs, we further consider two  numerical measurements: an average coverage rate of the 95\% prediction intervals and an average proper scores. Coverage rate is the proportion of the times that the interval contains the true value, and the proper score is the scoring rule by \cite{gneiting2007strictly}, which is an overall measure of the accuracy of the combined prediction mean and variance predictions. Specifically, the proper score has the following form:
$$
\text{proper score}=-\left(\frac{y-\mu_P}{\sigma_P}\right)^2-\log\sigma^2_P,
$$
where $y$ is the true output, $\mu_P$ is the predictive mean, and $\sigma^2_P$ is the predictive variance. A larger proper score indicates a better prediction. The results are summarized in Table \ref{tab:prediction_comparison}, which shows that the proposed method, \texttt{FIGP}, outperforms the two basis-expansion approaches in terms of both  predictions and uncertainty quantification. The average coverages of \texttt{FIGP} are close to the  nominal coverage 95\%, and the scores of \texttt{FIGP} are much higher than the two basis-expansion approaches.

\section{Applications to Inverse Scattering Problem}\label{sec:realcase}

In this section, we revisit the inverse scattering problem shown in Figure \ref{fig:realcase_real}. Let $D\subset \mathbb{R}^2$ denote an inhomogeneous isotropic scattering region of interest, and the functional input $g$, the support of which is $D$, is related to the refractive index for the region $D$ of the unknown scatterer. 
Given a set of finite element simulations as the training data, the goal of inverse scattering is to recover the functional input from a given far-field pattern. To achieve this goal, an important task is to construct a surrogate model for functional inputs.

In this study, 10 functional inputs, including $1,1+x_1,1-x_1,1+x_1x_2,1-x_1x_2,1+x_2,1+x_1^2,1-x_1^2,1+x_2^2$, and $1-x_2^2$, are conducted in the training set and the corresponding far-field patterns are shown in Figure \ref{fig:realcase_real}. 
Note that the inputs herein are given with explicit functional forms. In other applications where   discrete realizations of functions are available, the kernel functions can be numerically approximated by the discrete realizations as in \eqref{eq:linearkernelMC} and \eqref{eq:nonlinearkernelMC}.
The pre-processing of using principal component analysis (PCA) is first applied to reduce the dimension of the output images. The first three principal components, denoted by $\mathbf{u}_l\in\mathbb{R}^{1024},l=1,2,3$, are shown in Figure \ref{fig:realcase_pc}, which explain more than 99.99\% variability of the data. Therefore, given the functional input $g_i$ for $i=1,\ldots, 10$, the output of far-field images can be  approximated by $\sum^3_{l=1}f_l(g_i)\mathbf{u}_l$, where $f_1(g_i), f_2(g_i), f_3(g_i)$ are the first three PC scores. 

\begin{figure}[]
    \centering
    \includegraphics[width=\textwidth]{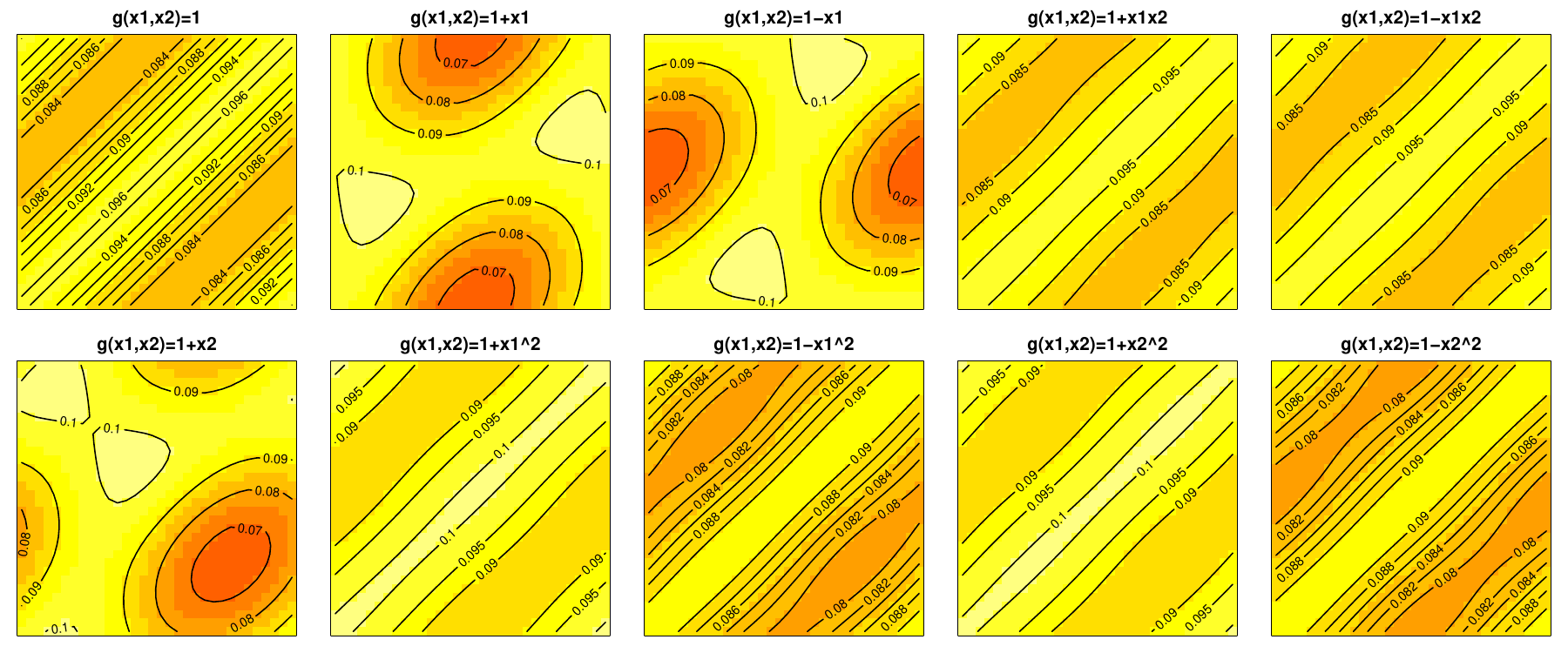}
    \caption{Training data in the application of the inverse scattering problem.}
    \label{fig:realcase_real}
\end{figure}

After the dimension reduction, the three-dimensional outputs $f_{1}(g),f_{2}(g)$ and $f_{3}(g)$ are assumed to be mutually independent and follow the functional-input GP as the surrogate model.
For any untried functional input $g\in V$, based on the results of \eqref{mean} and \eqref{var}, the far-field pattern can be predicted by the following normal distribution:  
$$
\mathcal{N}\left(\sum^3_{l=1}\mathbf{k}_l(g)^T \mathbf{K}^{-1}_{n,l} (\mathbf{f}_l-\mu_l\mathbf{1}_n)\mathbf{u}_l,\sum^3_{l=1}(K_l(g,g)-\mathbf{k}_l(g)^T \mathbf{K}^{-1}_{n,l} \mathbf{k}_l(g))\mathbf{u}^2_l\right),
$$
where $\mathbf{f}_l=(f_l(g_1),\ldots,f_l(g_{n}))$, $\mathbf{k}_l(g)=(K_l(g,g_1),\ldots,K_l(g,g_n))^T,(\mathbf{K}_{n,l})_{i,j}=K_l(g_i,g_j)$, and $K_l$ is the kernel function with the hyperparameters estimated based on  $\mathbf{f}_l$. 

The linear kernel of \eqref{eq:lk} and the nonlinear kernel of \eqref{eq:nonlinearkernel} are both performed. The optimal kernel is selected by comparing the LOOCV errors in predicting the far-field pattern. The LOOCV error based on the linear kernel is $3.6\times 10^{-6}$, which is smaller than that of the nonlinear kernel, $1.2\times 10^{-5}$. Therefore, we apply the linear kernel and examine its prediction performance for the test function, $g(\mathbf{x})=1-\sin(x_2)$.
Similar to Section \ref{sec:numericstudy}, two basis-expansion approaches, \texttt{FPCA} and \texttt{T3}, are compared. The images of the true far-field patterns and their predictions, along with their variances (in logarithm), are illustrated in Figure \ref{fig:realcase_prediction}.  Compared to the ground truth, the predictions of \texttt{FIGP} can capture the underlying structures reasonably well with some discrepancies appearing on the lower right corner. On the other hand, both of the predictions of \texttt{FPCA} and \texttt{T3} appear to deviate more from the ground truth. The MSEs and average scores are reported in Table \ref{tab:prediction_comparison_realcase}, which indicates that the proposed method can outperform the two basis-expansion approaches in terms of prediction accuracy and uncertainty quantification.


\begin{table}[!h]
\begin{center}
\begin{tabular}{ c|C{3cm}|C{3cm}|C{3cm}} 
 \toprule
Measurements & \texttt{FIGP} & \texttt{FPCA}& \texttt{T3} \\ 
 \midrule
 MSE &   $\boldsymbol{1.10\times 10^{-6}}$ & $1.07\times 10^{-4}$ & $9.06\times 10^{-5}$\\ 
Score & \textbf{12.13} & 6.89 & 6.39\\ 
 \bottomrule
\end{tabular}
\end{center}
    \caption{Prediction performance of the FIGP and basis-expansion approaches in the inverse scattering problem application (\texttt{FPCA} indicates an FPCA expansion approach and \texttt{T3} indicates the Taylor series expansion of degree 3), including MSEs and the average proper scores, in which the values with better performances are boldfaced.}
    \label{tab:prediction_comparison_realcase}
\end{table}

\section{Concluding Remarks}

Although GP surrogates are widely used in the analysis of complex system as an alternative to the direct analysis using computer experiments, most of the existing work are not applicable to the problems with functional inputs. To address this issue, two new types of kernel functions are introduced for functional-input GPs, including a linear and nonlinear kernel. Theoretical properties of the proposed surrogates, such as the convergence rate of the mean squared prediction error, are discussed. Numerical studies and the application to surrogate modeling in an inverse scattering problem show high prediction accuracy of the proposed method.

There are extensive studies in experimental design for conventional GP surrogate models, but the study of optimal designs for GPs with functional inputs remains scarce. In this paper, it is shown that space-fillingness is a desirable property in controlling the convergence rate of MSPE. An interesting topic for future research is to explore the construction procedure for efficient space-filling designs with functional inputs. 
Apart from experimental designs, another important future work is to explore systematic approaches to efficiently identify the functional input given an observed far-field pattern, which is the ultimate goal of inverse scattering problems. Based on the proposed GP surrogate, we will explore a Bayesian inverse framework that integrates computer experiments and real observations. Last but not the least, even though the numerical studies in Supplementary Material \ref{sec:samplepath} indicate that the smoothness parameter $\nu$ in the linear kernel function does not much affect the sample paths, it is worth exploring the theoretical properties regarding the choice of the parameter. We leave it for the future work.

\vspace{0.5cm}
\section*{Supplementary Material} 
Additional supporting materials can be found in Supplemental Materials, including the theoretical proofs of Propositions \ref{Prop1} and \ref{prop:pdofnl}, Theorems \ref{thm:linear} and \ref{coro:nlmatern},  and Corollaries \ref{coro:linear1} and \ref{coro:linear2}, the sample paths of the functional-input GP, the supporting tables and figures for Sections \ref{sec:numericstudy} and \ref{sec:realcase}, and the data, \textsf{R} code for reproducing the results in Sections
\ref{sec:numericstudy} and \ref{sec:realcase}.

\section*{Acknowledgments}
The authors gratefully acknowledge funding from NSF DMS-2113407, DMS-2107891, CCF-1934924, and NSFC 12101149.
\par

\bibliography{bib}

\def\spacingset#1{\renewcommand{\baselinestretch}%
{#1}\small\normalsize} \spacingset{1.3}

\newpage
\setcounter{page}{1}
\bigskip
\bigskip
\bigskip
\begin{center}
{\Large\bf Supplementary Materials for ``Functional-Input Gaussian Processes with Applications to Inverse Scattering Problems''}
\end{center}
\medskip

\setcounter{section}{0}
\setcounter{equation}{0}
\setcounter{figure}{0}
\setcounter{table}{0}
\def\theequation{S\arabic{section}.\arabic{equation}}
\def\thesection{S\arabic{section}}
\def\thefigure{S\arabic{figure}}
\def\thetable{S\arabic{table}}

\section{Proof of Proposition \ref{Prop1}}

For any non-zero $(\alpha_1,...,\alpha_n)$ and $(g_1,...,g_n)$, it follows the quadratic form:
\begin{align*}
    & \sum_{j=1}^n\sum_{k=1}^n \alpha_j\alpha_kK(g_j,g_k) = \sum_{j=1}^n\sum_{k=1}^n \alpha_j\alpha_k\int_{\Omega}\int_{\Omega} g_j(\mathbf{x})g_k(\mathbf{x}')\Psi(\mathbf{x},\mathbf{x}'){\rm{d}}\mathbf{x}{\rm{d}}\mathbf{x}'\nonumber\\
    = & \int_{\Omega}\int_{\Omega}\sum_{j=1}^n\sum_{k=1}^n \alpha_j\alpha_k g_j(\mathbf{x})g_k(\mathbf{x}')\Psi(\mathbf{x},\mathbf{x}'){\rm{d}}\mathbf{x}{\rm{d}}\mathbf{x}'
    \nonumber\\
    = & \int_{\Omega}\int_{\Omega}\left(\sum_{j=1}^n\alpha_j g_j(\mathbf{x})\right)^2\Psi(\mathbf{x},\mathbf{x}'){\rm{d}}\mathbf{x}{\rm{d}}\mathbf{x}'\geq 0,
\end{align*}
and the quadratic form is strictly greater than zero if $g_1,...,g_n$ are linearly independent. This finishes the proof.

\section{Proof of Proposition \ref{Proplinear}}
For any $a,b\in \mathbb{R}$ and $g_1,g_2\in V$, it follows that 
\begin{align*}
    f(ag_1+bg_2) = & \sum_{j=1}^\infty \sqrt{\lambda_j}\langle \phi_j,ag_1+bg_2 \rangle_{L_2(\Omega)}Z_j\nonumber\\
    =&  \sum_{j=1}^\infty \sqrt{\lambda_j}\left(a\langle \phi_j,g_1 \rangle_{L_2(\Omega)}+b\langle \phi_j,g_2 \rangle_{L_2(\Omega)}\right)Z_j\nonumber\\
    = & af(g_1)+bf(g_2),
\end{align*}
which finishes the proof.

\section{Proof of Theorem \ref{thm:linear}}
In order to prove the theorem, the following lemma is provided, which can be found in Proposition 10.28 of \cite{suppwendland2004scattered}.

\begin{lemma}\label{lem:p28ofwendland}
Suppose $\Psi$ is a symmetric and positive definite kernel on $\Omega$. Then the integral operator $\mathcal{T}$ maps $L_2(\Omega)$ continuously into the reproducing kernel Hilbert space $\mathcal{N}_\Psi(\Omega)$. It is the adjoint of the embedding operator of the reproducing kernel Hilbert space $\mathcal{N}_\Psi(\Omega)$ into $L_2(\Omega)$, i.e., it satisfies
\begin{align*}
    \langle g,v\rangle_{L_2(\Omega)} = \langle g,\mathcal{T}v\rangle_{\mathcal{N}_\Psi(\Omega)}, g\in \mathcal{N}_\Psi(\Omega), v\in L_2(\Omega).
\end{align*}
\end{lemma}

Now we are ready to prove Theorem \ref{thm:linear}. For any $\mathbf{u}_n=(u_1,...,u_n)^T\in \mathbb{R}^n$, it follows that 
\begin{align}\label{eq:pflimspe1}
    & \mathbb{E}\left(f(g)-\sum_{j=1}^n u_j f(g_j)\right)^2 =  K(g,g)-2\sum_{j=1}^n u_j K(g,g_j)+ \sum_{j=1}^n \sum_{l=1}^nu_ju_l K(g_j,g_l)\nonumber\\
    = & \int_\Omega\int_\Omega g(\mathbf{x})g(\mathbf{x}')\Psi(\mathbf{x},\mathbf{x}'){\rm{d}}\mathbf{x}{\rm{d}}\mathbf{x}' - 2\sum_{j=1}^n u_j\int_\Omega\int_\Omega g(\mathbf{x})g_j(\mathbf{x}')\Psi(\mathbf{x},\mathbf{x}'){\rm{d}}\mathbf{x}{\rm{d}}\mathbf{x}' \nonumber\\
    & + \sum_{j=1}^n \sum_{l=1}^nu_ju_l\int_\Omega\int_\Omega g_j(\mathbf{x})g_l(\mathbf{x}')\Psi(\mathbf{x},\mathbf{x}'){\rm{d}}\mathbf{x}{\rm{d}}\mathbf{x}'\nonumber\\
    = & \langle g,\mathcal{T}g\rangle_{L_2(\Omega)} - 2\sum_{j=1}^n u_j\langle g,\mathcal{T}g_j\rangle_{L_2(\Omega)} + \sum_{j=1}^n \sum_{l=1}^nu_ju_l\langle g_j,\mathcal{T}g_l\rangle_{L_2(\Omega)}\nonumber\\
    = & \langle \mathcal{T}g,\mathcal{T}g\rangle_{\mathcal{N}_\Psi(\Omega)} - 2\sum_{j=1}^n u_j\langle \mathcal{T}g,\mathcal{T}g_j\rangle_{\mathcal{N}_\Psi(\Omega)} + \sum_{j=1}^n \sum_{l=1}^nu_ju_l\langle \mathcal{T}g_j,\mathcal{T}g_l\rangle_{\mathcal{N}_\Psi(\Omega)}\nonumber\\
    = & \left\|\mathcal{T}g - \sum_{j=1}^n u_j \mathcal{T}g_j\right\|_{\mathcal{N}_\Psi(\Omega)}^2,
\end{align}
where the third equality is by Lemma \ref{lem:p28ofwendland}.

Since $\hat {\mathbf{u}}_n:= \mathbf{K}_n^{-1}\mathbf{k}_n(g)$ minimizes the MSPE, it also minimizes $\left\|\mathcal{T}g - \sum_{j=1}^n u_j \mathcal{T}g_j\right\|_{\mathcal{N}_\Psi(\Omega)}^2$, where $\mathbf{K}_n$ and $\mathbf{k}_n(g)$ are as in \eqref{mean}. By taking minimum on both sides of \eqref{eq:pflimspe1}, we obtain 
\begin{align}\label{eq:pflimspe3}
     \mathbb{E}\left(f(g)-\hat f(g)\right)^2
    = \min_{\mathbf{u}\in \mathbb{R}^n} \left\|\mathcal{T}g - \sum_{j=1}^n u_j \mathcal{T}g_j\right\|_{\mathcal{N}_\Psi(\Omega)}^2,
\end{align}
which finishes the proof.

\section{Proof of Corollary \ref{coro:linear1}}
In order to prove this corollary, we need the following lemma, which states the asymptotic rates of the eigenvalues of $K$. Lemma \ref{lemDecayEig} is implied by the proof of Lemma 18 of \cite{supptuo2020kriging}. In the rest of the Supplementary Materials, we will use the following notation. For two positive sequences $a_n$ and $b_n$, we write $a_n\asymp b_n$ if, for some constants $C,C'>0$, $C\leq a_n/b_n \leq C'$. For notational simplicity, we will use $C,C',C_1,C_2,...$ to denote the constants, of which the values can change from line to line.

\begin{lemma}\label{lemDecayEig}
Let $\Psi$ be a Mat\'ern kernel function  defined in \eqref{matern2}, and $\lambda_1\geq \lambda_2 \geq ...>0$ be its eigenvalues. Then, $\lambda_k\asymp k^{-2\nu/d}$.
\end{lemma}

\begin{proof}[Proof of Corollary \ref{coro:linear1}]

By \eqref{eq:mercer1}, we have
$
    \mathcal{T}g(\mathbf{x}) =  \int_\Omega g(\mathbf{x}')\Psi(\mathbf{x},\mathbf{x}'){\rm d}\mathbf{x}'=\sum_{j=1}^\infty \lambda_j\langle g,\phi_j\rangle_{L_2(\Omega)}\phi_j(\mathbf{x}).
$
Therefore, 
\begin{align}\label{eq:idofTmap}
    \left\|\mathcal{T}g - \sum_{j=1}^n u_j \mathcal{T}g_j\right\|_{\mathcal{N}_\Psi(\Omega)}^2 & =  \left\|\sum_{j=1}^\infty \lambda_j\langle g-\sum_{k=1}^n u_k g_k,\phi_j\rangle_{L_2(\Omega)}\phi_j\right\|_{\mathcal{N}_\Psi(\Omega)}^2\nonumber\\
    &= \sum_{j=1}^\infty \lambda_j\langle g-\sum_{k=1}^n u_k g_k,\phi_j\rangle_{L_2(\Omega)}^2,
\end{align}
where the last equality holds by Theorem 10.29 of \cite{suppwendland2004scattered}.
Recall that $g_k=\phi_k$ for $k=1,...,n$. Now take $u_k=\langle g,\phi_k\rangle_{L_2(\Omega)}$ for $k=1,...,n$. 
It follows from Theorem \ref{thm:linear} and \eqref{eq:idofTmap} that
\begin{align}\label{eq:mspeeg1eq1}
     & \mathbb{E}\left(f(g)-\hat f(g)\right)^2 = \min_{\mathbf{u}\in \mathbb{R}^n} \left\|\mathcal{T}g - \sum_{j=1}^n u_j \mathcal{T}g_j\right\|_{\mathcal{N}_\Psi(\Omega)}^2 \nonumber\\
     \leq & \sum_{j=1}^\infty \lambda_j\langle g-\sum_{k=1}^n u_k g_k,\phi_j\rangle_{L_2(\Omega)}^2= \sum_{j=n+1}^\infty \lambda_j\langle g,\phi_j\rangle_{L_2(\Omega)}^2.
\end{align}
This indicates that the MSPE depends on the tail behavior of the summation $\sum_{j=1}^\infty \lambda_j\langle g,\phi_j\rangle_{L_2(\Omega)}^2$. Because $\Psi$ is a Mat\'ern kernel function defined in \eqref{matern2}, Lemma \ref{lemDecayEig} implies $\lambda_j\asymp j^{-\frac{2\nu}{d}}$. Then, an explicit convergence rate can be obtained via 
\begin{align*}
    \mathbb{E}\left(f(g)-\hat f(g)\right)^2\leq \sum_{j=n+1}^\infty \lambda_j\langle g,\phi_j\rangle_{L_2(\Omega)}^2 \leq \lambda_n\sum_{j=n+1}^\infty \langle g,\phi_j\rangle_{L_2(\Omega)}^2 \leq C_1\|g\|_{L_2}^2 n^{-\frac{2\nu}{d}},
\end{align*}
where we utilize $\sum_{j=n+1}^\infty \langle g,\phi_j\rangle_{L_2(\Omega)}^2\leq \sum_{j=1}^\infty \langle g,\phi_j\rangle_{L_2(\Omega)}^2 = \|g\|_{L_2}^2$. This finishes the proof of \eqref{eq:l1c11}.

Next, we consider the case $g\in \mathcal{N}_{\Psi}(\Omega)$. Theorem 10.29 of \cite{suppwendland2004scattered} yields that
\begin{align*}
    \|g\|_{\mathcal{N}_\Psi(\Omega)}^2 = \sum_{j=1}^\infty \frac{\langle g,\phi_j\rangle_{L_2(\Omega)}^2}{\lambda_j}<\infty.
\end{align*}
Then an alternative way to bound \eqref{eq:mspeeg1eq1} is by
\begin{align}\label{eq:mspeeg1eq2}
    & \mathbb{E}\left(f(g)-\hat f(g)\right)^2\leq \sum_{j=n+1}^\infty \lambda_j\langle g,\phi_j\rangle_{L_2(\Omega)}^2 = \sum_{j=n+1}^\infty \lambda_j^2\frac{\langle g,\phi_j\rangle_{L_2(\Omega)}^2}{\lambda_j}\nonumber\\ 
    \leq & \lambda_n^2\sum_{j=n+1}^\infty \frac{\langle g,\phi_j\rangle_{L_2(\Omega)}^2}{\lambda_j}
    \leq \lambda_n^2\|g\|_{\mathcal{N}_\Psi(\Omega)}^2\leq C_2\|g\|_{\mathcal{N}_\Psi(\Omega)}^2n^{-\frac{4\nu}{d}}.
\end{align}
This finishes the proof of \eqref{eq:l1c12}, and thus the proof of Corollary \ref{coro:linear1}.
\end{proof}

\section{Proof of Corollary \ref{coro:linear2}}
The following lemma is used in the proof.

\begin{lemma}[\citealp{suppwu1993local}, Theorem 5.14]\label{Th:Matern}
	Let $\Omega$ be compact and convex with a positive Lebesgue measure; $\Psi(\mathbf{x},\mathbf{x}')$ be a Mat\'ern kernel given by \eqref{matern2} with the smoothness parameter $\nu$. Then there exist constants $c,c_0$ depending only on $\Omega,\nu$ and the lengthscale parameter $\Theta$ in \eqref{matern2}, such that
	$\Psi(\mathbf{x},\mathbf{x})-\mathbf{r}_n(\mathbf{x})^T\mathbf{R}_n^{-1}\mathbf{r}_n(\mathbf{x})\leq c h_{\mathbf{X}_n,\Omega}^{2\nu}$ provided that $h_{\mathbf{X}_n,\Omega}\leq c_0$, where $\mathbf{R}_n=(\Psi(\mathbf{x}_j,\mathbf{x}_k))_{jk}$ and $\mathbf{r}_n(\mathbf{x}) = (\Psi(\mathbf{x},\mathbf{x}_1),...,\Psi(\mathbf{x},\mathbf{x}_n))^T$.
\end{lemma}

\begin{proof}[Proof of Corollary \ref{coro:linear2}]

For any $\mathbf{u}_n=(u_1,...,u_n)^T$, by \eqref{eq:idofTmap}, we have
\begin{align}\label{eq:eg2eq1}
    & \left\|\mathcal{T}g - \sum_{j=1}^n u_j \mathcal{T}g_j\right\|_{\mathcal{N}_\Psi(\Omega)}^2 =\sum_{j=1}^\infty \lambda_j\langle g-\sum_{k=1}^n u_k g_k,\phi_j\rangle_{L_2(\Omega)}^2\nonumber\\
    = & \sum_{j=1}^\infty \lambda_j\langle g-\sum_{k=1}^n u_k \Psi(\mathbf{x}_k,\cdot),\phi_j\rangle_{L_2(\Omega)}^2 \leq C\sum_{j=1}^\infty \langle g-\sum_{k=1}^n u_k \Psi(\mathbf{x}_k,\cdot),\phi_j\rangle_{L_2(\Omega)}^2 \nonumber\\
     = & C\left\|g-\sum_{k=1}^n u_k \Psi(\mathbf{x}_k,\cdot)\right\|_{L_2(\Omega)}^2 \leq C_1\sup_{\mathbf{x}\in \Omega}\left|g(\mathbf{x})-\sum_{k=1}^n u_k \Psi(\mathbf{x}_k,\mathbf{x})\right|^2,
\end{align}
where the last equality holds because $\phi_j$'s are orthogonal basis in $L_2(\Omega)$. Therefore, we can take $\mathbf{u}_n = \mathbf{R}_n^{-1}\mathbf{g}_n$, where  $\mathbf{g}_n=(g(\mathbf{x}_1),...,g(\mathbf{x}_n))^T$. Then the term, $\sup_{\mathbf{x}\in \Omega}\left|g(\mathbf{x})-\sum_{k=1}^n u_k \Psi(\mathbf{x}_k,\mathbf{x})\right|$, becomes the prediction error of the radial basis function interpolation, which is well established in the literature \citep{suppwendland2004scattered}.

For the completeness of this proof, we include the proof of an upper bound here. Let $v_k(\mathbf{x}) = (\mathbf{R}_n^{-1}\mathbf{r}_n(\mathbf{x}))_k$ for $k=1,\ldots,n$. For any $\mathbf{x}\in \Omega$, the reproducing property implies that
\begin{align}\label{eq:eg2eq2}
    & \left|g(\mathbf{x})-\sum_{k=1}^n u_k \Psi(\mathbf{x}_k,\mathbf{x})\right|^2 = \left|\langle g,\Psi(\mathbf{x},\cdot)\rangle_{\mathcal{N}_{\Psi}(\Omega)}-\sum_{k=1}^n v_k(\mathbf{x})g(\mathbf{x}_k)\right|^2\nonumber\\
    = & \left|\langle g,\Psi(\mathbf{x},\cdot)\rangle_{\mathcal{N}_{\Psi}(\Omega)}-\sum_{k=1}^n v_k(\mathbf{x}) \langle g,\Psi(\mathbf{x}_k,\cdot)\rangle_{\mathcal{N}_{\Psi}(\Omega)}\right|^2\nonumber\\
    \leq & \|g\|_{\mathcal{N}_{\Psi}(\Omega)}^2\left\|\Psi(\mathbf{x},\cdot)-\sum_{k=1}^n v_k(\mathbf{x})\Psi(\mathbf{x}_k,\cdot)\right\|_{\mathcal{N}_{\Psi}(\Omega)}^2\nonumber\\
    = & \|g\|_{\mathcal{N}_{\Psi}(\Omega)}^2(\Psi(\mathbf{x},\mathbf{x})-\mathbf{r}_n(\mathbf{x})^T\mathbf{R}_n^{-1}\mathbf{r}_n(\mathbf{x})),
\end{align}
where the inequality is by the Cauchy-Schwarz inequality. A bound on $\Psi(\mathbf{x},\mathbf{x})-\mathbf{r}_n(\mathbf{x})^T\mathbf{R}_n^{-1}\mathbf{r}_n(\mathbf{x})$ can be obtained via Lemma \ref{Th:Matern}, which gives 
\begin{align}\label{eq:pwub}
    \Psi(\mathbf{x},\mathbf{x})-\mathbf{r}_n(\mathbf{x})^T\mathbf{R}_n^{-1}\mathbf{r}_n(\mathbf{x})\leq C_2h_{\mathbf{X}_n,\Omega}^{2\nu},
\end{align}
where $h_{\mathbf{X}_n,\Omega}$ is the fill distance of $\mathbf{X}_n$. Combining \eqref{eq:thmlieq} with \eqref{eq:eg2eq1}, \eqref{eq:eg2eq2} and \eqref{eq:pwub}, we obtain that 
\begin{align*}
    \mathbb{E}\left(f(g)-\hat f(g)\right)^2 \leq C_3h_{\mathbf{X}_n,\Omega}^{2\nu},
\end{align*}
which finishes the proof.
\end{proof}

\section{Proof of Proposition \ref{prop:pdofnl}}
Without loss of generality, we assume $\gamma=1$. It suffices to show that $\mathbf{K}_n=(K(g_j,g_k))_{jk}$ is positive definite for any $g_1,...,g_n\in V$, which can be done by showing that there exist $n$ points $\mathbf{a}_1,...,\mathbf{a}_n\in \mathbb{R}^m$ with $m < \infty$ such that $\|\mathbf{a}_j-\mathbf{a}_k\|_2 = \|g_j-g_k\|_{L_2(\Omega)}$. Let $V_n = {\rm span}(\{g_1,...,g_n\})$, which is the linear space spanned by $g_1,...,g_n$. Clearly, $V_n$ is a finite dimensional space. Let $\phi_1,...,\phi_m$ be an orthogonal basis of $V_n$ with respect to $L_2(\Omega)$, with $m\leq n$, which can be found via the Gram–Schmidt process. Thus, given this basis, each $g_j$ can be written as
\begin{align*}
    g_j = \sum_{k=1}^m a_{jk}\phi_k
\end{align*}
with $\mathbf{a}_j = (a_{j1},...,a_{jm})^T\in \mathbb{R}^m$. Then it can be verified that $\|\mathbf{a}_j-\mathbf{a}_k\|_2 = \|g_j-g_k\|_{L_2(\Omega)}$. Since $\mathbf{K}_n=(K(g_j-g_k))_{jk}=(\psi(\|g_j-g_k\|_{L_2(\Omega)}))_{jk}=(\psi(\|\mathbf{a}_j-\mathbf{a}_k\|_2))_{jk}$, which is positive definite, this finishes the proof.



\section{Proof of Theorem \ref{coro:nlmatern}}

We first provide a characterization on the function class $V$. Since $\Phi$ is a positive definite function, we can apply Mercer's theorem to $\Phi$ and obtain
\begin{align}\label{eq:mercer}
\Phi(\mathbf{x},\mathbf{x}') = \sum_{j=1}^\infty \lambda_{\Phi,j} \phi_j(\mathbf{x})\phi_j(\mathbf{x}'), \quad \mathbf{x},\mathbf{x}'\in \Omega,
\end{align}
where $\lambda_{\Phi,1}\geq \lambda_{\Phi,2}\geq...> 0$ are the eigenvalues and $\{\phi_k\}_{k\in\mathbb{N}}$ are the eigenfunctions, and the summation is uniformly and absolutely convergent. Because $g\in V\subset \mathcal{N}_\Phi(\Omega)$, by Theorem 10.29 of \cite{suppwendland2004scattered}, the summation
\begin{align*}
    \|g\|_{\mathcal{N}_\Phi(\Omega)}^2 = \sum_{j=1}^\infty \lambda_{\Phi,j}^{-1}\langle g,\phi_j\rangle_{L_2(\Omega)}^2 \leq 1.
\end{align*}

\begin{definition}\label{def:icc}
	A set $\Omega\subset{\mathbb{R}^d}$ is said to satisfy an interior cone condition if there exists an angle $\alpha\in (0,\pi/2)$ and a radius $\mathfrak{r}>0$ such that for every $\mathbf{x}\in\Omega$, a unit vector $\xi(\mathbf{x})$ exists such that the cone
	$$C(\mathbf{x},\xi(\mathbf{x}),\alpha,\mathfrak{r}):=\left\{\mathbf{x}+\eta \tilde{\mathbf{x}}:\tilde{\mathbf{x}}\in\mathbb{R}^d,\|\tilde{\mathbf{x}}\|=1,\tilde{\mathbf{x}}^T\xi(\mathbf{x})\geq \cos(\alpha),\eta\in[0,\mathfrak{r}]\right\} $$
	is contained in $\Omega$.
\end{definition}

We need the following lemma, which is Theorem 11.8 of \cite{suppwendland2004scattered}; also see Theorem 3.14 of \cite{suppwendland2004scattered}. Lemma \ref{lem:lpr} ensures the existence of the local polynomial reproduction. 

\begin{lemma}[Local polynomial reproduction]\label{lem:lpr}
Let $l\in\mathbb{N}_0$ and $\pi_l(\mathbb{R}^d)$ be the set of $d$-variate polynomials with absolute degree no more than $l$. Suppose $\Omega\subset \mathbb{R}^d$ is bounded and satisfies an interior cone condition. Define
\begin{align*}
    C_1=2, C_2=\frac{16(1+\sin(\alpha))^2l^2}{3\sin^2(\alpha)}, c_0=\frac{\mathfrak{r}}{C_2},
\end{align*}
where $\alpha$ and $\mathfrak{r}$ are defined in Definition \ref{def:icc}. Then for all $\mathbf{X}_n=\{\mathbf{x}_1,\ldots,\mathbf{x}_n\}\subset \Omega$ with $h_{\mathbf{X}_n,\Omega}\leq c_0$ and every $\mathbf{x}\in\Omega$ there exist numbers $\tilde{u}_1(\mathbf{x}),\ldots,\tilde{u}_n(\mathbf{x})$ with

(1) $\sum_{j=1}^n \tilde{u}_j(\mathbf{x}) p(\mathbf{x}_j)=p(\mathbf{x}) $ {for all} $p\in\pi_l(\mathbb{R}^d)$,

(2) $\sum_{j=1}^n|\tilde{u}_j(\mathbf{x})|\leq C_1$,

(3) $\tilde{u}_j(\mathbf{x})=0$, {if} $\|\mathbf{x}-\mathbf{x}_j\|> C_2h_{\mathbf{X}_n,\Omega}$.
\end{lemma}

The following lemma provides an upper bound of the MSPE using $\tilde u_j(\mathbf{x})$, which can be found from the proof of Theorem 11.9 and (11.6) of \cite{suppwendland2004scattered}.


\begin{lemma}\label{lem:bofpow}
Suppose $\Phi=r(\|\cdot\|_2)\in C^k(\mathbb{R}^d)$ is positive definite. Let $\Omega$ be a compact region satisfying an interior cone condition. Then for $h_{\mathbf{X}_n,\Omega}\leq h_0$, 
\begin{align*}
    \Phi(\mathbf{x},\mathbf{x}) - 2\sum_{j=1}^n \tilde u_j\Phi(\mathbf{x},\mathbf{x}_j) + \sum_{j=1}^n\sum_{k=1}^n \tilde u_j\tilde u_k\Phi(\mathbf{x}_j,\mathbf{x}_k)\leq (1+C_1)^2 \max_{0\leq s\leq 2C_2h_{\mathbf{X}_n,\Omega}} |\phi(s)-p(s^2)|,
\end{align*}
where $p\in \pi_{\lfloor l/2 \rfloor}(\mathbb{R})$.
\end{lemma}

\begin{proof}[Proof of Theorem \ref{coro:nlmatern}]
Define a map $h:V\rightarrow W$ between the function class $V$ and the set $W=\{\mathbf{a}=(a_1,...,a_n,...)^T:\sum_{j=1}^\infty \lambda_{\Phi,j}^{-1}a_j^2\leq 1\}\subset l_2(\mathbb{R}^\infty)$ as 
\begin{align*}
    h(g)=(\langle g,\phi_1\rangle_{L_2(\Omega)},...,\langle g,\phi_n\rangle_{L_2(\Omega)},...)^T.
\end{align*}
It can be verified that $\|g\|_{L_2(\Omega)}=\|h(g)\|_2$. Therefore, we can define a new positive definite function $K_1$ on $W$ which satisfies
\begin{align}
    K_1(\mathbf{a},\mathbf{a}')=\psi(\|\mathbf{a}-\mathbf{a}'\|_2)= \psi(\|g-g'\|_{L_2(\Omega)})=K(g,g'), \quad\forall g,g'\in V,
\end{align}
where $\mathbf{a}=h(g)$ and $\mathbf{a}'=h(g')$. Define $\mathbf{a}^{(j)}=h(g_j)$. For any $\mathbf{u}_n=(u_1,...,u_n)^T\in \mathbb{R}^n$, it follows that 
\begin{align}\label{eq:pfnlimspe1}
    & \mathbb{E}\left(f(g)-\sum_{j=1}^n u_j f(g_j)\right)^2\nonumber\\
    = & K(g,g)-2\sum_{j=1}^n u_j K(g,g_j)+ \sum_{j=1}^n \sum_{k=1}^nu_ju_k K(g_j,g_k)\nonumber\\
    = & \psi(0) - 2\sum_{j=1}^n u_j \psi(\|\mathbf{a}-\mathbf{a}^{(j)}\|_2)+ \sum_{j=1}^n \sum_{k=1}^nu_ju_l \psi(\|\mathbf{a}^{(j)}-\mathbf{a}^{(k)}\|_2).
\end{align}
Let $\mathbf{b}_j = (a_1^{(j)},...,a_m^{(j)})^T$, $\mathbf{b} = (a_1,...,a_m)^T$, $\mathbf{a}_c^{(j)} = (a_{m+1}^{(j)},a_{m+2}^{(j)},...)$, and $\mathbf{a}_c= (a_{m+1},a_{m+2},...)$, where $m$ will be determined later. Then $\mathbf{a}^{(j)} = (\mathbf{b}_j^T,(\mathbf{a}_c^{(j)})^T)^T$ and $\mathbf{a} = (\mathbf{b}^T,\mathbf{a}_c^T)^T$. Applying Lemma \ref{lem:lpr} to \eqref{eq:pfnlimspe1}, we obtain that for some $\tilde{u}_j$, 
\begin{align}\label{eq:pfthmnlut1}
    & \sum_{j=1}^n \tilde{u}_jp(\mathbf{b}_j) = p(\mathbf{b}), \mbox{ for all }p\in \pi_l(\mathbb{R}^m), 
    & \sum_{j=1}^n|\tilde{u}_j(\mathbf{b})|\leq C_1, \mbox{ and }
    & \tilde{u}_j(\mathbf{b})=0, \mbox{ if }\|\mathbf{b}-\mathbf{b}_j\|_2> C_2h_{\mathbf{B}_n,\Omega},
\end{align}
when $h_{\mathbf{B}_n,\Omega}\leq c_0$, where $\mathbf{B}_n=\{\mathbf{b}_1,...,\mathbf{b}_n\}$. Note that $C_2$ and $c_0$ depend on the interior cone condition and $l$. In particular, they change as the dimension of $\mathbf{b}$ and the degree $l$ of polynomials $p$ in \eqref{eq:pfthmnlut1} change. 

Since $\mathbf{a},\mathbf{a}^{(j)}\in W$, it follows that $\mathbf{a}\leq \sqrt{\lambda_{\Phi,j}}$ and $\mathbf{a}^{(j)}\leq \sqrt{\lambda_{\Phi,j}}$. 
Define a set $V_2=\bigtimes_{k=1}^m [0,\sqrt{\lambda_{\Phi,k}}]$. It can be verified that $\mathbf{b},\mathbf{b}_j\in V_2$. Set $\alpha = \pi /6$ and $\mathfrak{r}=\sqrt{\lambda_{\Phi,m}}/2$. It can be verified that the interior cone condition is satisfied. Then $C_2=48l^2$ and $c_0=\frac{\sqrt{\lambda_{\Phi,m}}}{96l^2}$. 

With $\tilde{u}_j$ defined in \eqref{eq:pfthmnlut1}, by \eqref{eq:pfnlimspe1}, it follows that
\begin{align}\label{eq:pflimspe2}
    & \mathbb{E}\left(f(g)-\hat f(g)\right)^2\nonumber\\
    \leq & \psi(0) - 2\sum_{j=1}^n \tilde{u}_j \psi(\|\mathbf{a}-\mathbf{a}^{(j)}\|_2)+ \sum_{j=1}^n \sum_{l=1}^n\tilde{u}_j\tilde{u}_l \psi(\|\mathbf{a}^{(j)}-\mathbf{a}^{(l)}\|_2)\nonumber\\
    = & \left(\psi(0) - 2\sum_{j=1}^n \tilde{u}_j \psi(\|\mathbf{b}-\mathbf{b}_j\|_2)+ \sum_{j=1}^n \sum_{k=1}^n\tilde{u}_j\tilde{u}_k \psi(\|\mathbf{b}_j-\mathbf{b}_k\|_2)\right)\nonumber\\
    & +\left(- 2\sum_{j=1}^n \tilde{u}_j \left(\psi(\|\mathbf{a}-\mathbf{a}^{(j)}\|_2)-\psi(\|\mathbf{b}-\mathbf{b}_j\|_2)\right)+ \sum_{j=1}^n \sum_{k=1}^n\tilde{u}_j\tilde{u}_k \left(\psi(\|\mathbf{a}^{(j)}-\mathbf{a}^{(l)}\|_2)-\psi(\|\mathbf{b}_j-\mathbf{b}_k\|_2)\right)\right)\nonumber\\
    := & I_1 + I_2. 
\end{align}
The first term can be bounded by Lemma \ref{lem:bofpow}, which gives  
\begin{align}\label{eq:pflimsI1}
    I_1 \leq &  9 \max_{0\leq s\leq 2C_2h_{\mathbf{B}_n,\Omega}} |\psi(s)-p(s^2)| =  9 \max_{0\leq s\leq \sqrt{\lambda_{\Phi,m}}} |\psi(s)-p(s^2)|,
\end{align}
for some $p\in \pi_{\lfloor l/2 \rfloor}(\mathbb{R})$, provided $h_{\mathbf{B}_n,\Omega}\leq c_0$. Since $\|\mathbf{b}_j-\mathbf{b}\|_2\leq \|\mathbf{a}-\mathbf{a}^{(j)}\|_2 = \|g-g_j\|_{L_2(\Omega)}$, we have $h_{\mathbf{B}_n,\Omega}\leq h_{G_n,V}$ so  $h_{\mathbf{B}_n,\Omega}\leq c_0$ holds.

Next, we consider bounding $I_1$ with a Mat\'ern kernel function $\psi$.
Lemma \ref{lemDecayEig} implies that $\lambda_{\Phi,j}\asymp j^{-\frac{2\nu_1}{d}}$. 
By the expansion of modified Bessel function \citep{suppabramowitz1948handbook}, $\psi$ can be written as
\begin{align*}
    \psi(r) = \sum_{k=0}^{\lfloor \nu \rfloor} c_k r^{2k} + c_\psi(r),
\end{align*}
where 
\begin{align*}
    c_\psi(r) = \left\{
    \begin{array}{cc}
        cr^{2\nu}\log r + O(r^{2\nu}) &  \nu = 1,2,...\\
        cr^{2\nu} + O(r^{2(\lfloor\nu\rfloor+1)}) & \mbox{ otherwise.}
    \end{array}
    \right.
\end{align*}
Therefore, we can take $p(s^2) = -\sum_{k=0}^{\lfloor \nu \rfloor} c_k s^{2k}$ and obtain that 
\begin{align}\label{eq:coronlmI1}
    \max_{0\leq s\leq \sqrt{\lambda_{\Phi,m}}}|\psi(s)-p(s^2)| \leq \left\{
    \begin{array}{ll}
        C_2\lambda_{\Phi,m}^{\nu}\log (\lambda_{\Phi,m}^{-1}) \leq C_3 m^{-\frac{2\nu\nu_1}{d}}\log m, &  \nu = 2,...\\
        C_4\lambda_{\Phi,m}^{\nu}\leq C_5 m^{-\frac{2\nu\nu_1}{d}}, & \mbox{otherwise.}
    \end{array}
    \right.
\end{align}
By \eqref{eq:pflimsI1}, 
\begin{align}\label{eq:pflimsI1X}
    I_1 \leq 9\max_{0\leq s\leq \sqrt{\lambda_{\Phi,m}}} |\psi(s)-p(s^2)|\leq \left\{
    \begin{array}{ll}
        C_6 m^{-\frac{2\nu\nu_1}{d}}\log m, &  \nu = 2,...\\
        C_7 m^{-\frac{2\nu\nu_1}{d}}, & \mbox{otherwise.}
    \end{array}
    \right.
\end{align}
It remains to bound $I_2$. For a Mat\'ern kernel function $\psi$, 
it can be verified that for all $s_1,s_2\in [0,s]$,
\begin{align*}
    |\psi(s_1)-\psi(s_2)|\leq C_8|s_1-s_2|^{2\tau}, 
\end{align*}
where $\tau=\min(\nu,1)$. Therefore, we can rewrite \eqref{eq:pflimspe2} as
\begin{align}\label{eq:coronlmI2}
    |I_2|\leq &  2C_8\sum_{j=1}^n \tilde{u}_j |\|\mathbf{a}-\mathbf{a}^{(j)}\|_2 - \|\mathbf{b}-\mathbf{b}_j\|_2|^{2\tau} + C_8\sum_{j=1}^n \sum_{k=1}^n\tilde{u}_j\tilde{u}_k |\|\mathbf{a}^{(j)}-\mathbf{a}^{(k)}\|_2-\|\mathbf{b}_j-\mathbf{b}_k\|_2|^{2\tau}\nonumber\\
    \leq & 2C_8\sum_{j=1}^n \left| \tilde{u}_j\right| \|\mathbf{a}-\mathbf{a}^{(j)} - (\mathbf{b}-\mathbf{b}_j)\|_2^{2\tau} + C_8\sum_{j=1}^n \sum_{k=1}^n\left|\tilde{u}_j\right|\left|\tilde{u}_l\right| \|\mathbf{a}^{(j)}-\mathbf{a}^{(l)} - (\mathbf{b}_j-\mathbf{b}_k)\|_2^{2\tau}\nonumber\\
    \leq & C_8(4C_1 + 2C_1^2)\lambda_{\Phi,m+1}^{\tau}\leq C_8(4C_1 + 2C_1^2)m^{-\frac{2\tau\nu_1}{d}}.
\end{align}
Because $m^{-\frac{2\nu\nu_1}{d}}\log m\leq m^{-\frac{2\nu_1\tau}{d}}\log m$ and $\log m>1$, combining \eqref{eq:coronlmI1}, \eqref{eq:coronlmI2} and \eqref{eq:pflimspe2} leads to 
\begin{align}\label{eq:thmnonlieqXX}
    \mathbb{E}\left(f(g)-\hat f(g)\right)^2\leq C_9m^{-\frac{2\nu_1\tau}{d}}\log m.
\end{align}
The last step is to compute $m$ such that there exist $n$ functions, $h_{G_n,V}\leq C_0m^{-\frac{2\nu_1}{d}}$. Since it is known that a unit ball of $\mathcal{N}_\Phi(\Omega)$ has a covering number
\begin{align*}
    N(\delta,V,\|\cdot\|_{L_\infty})\leq C_{10}\exp(C_{11}\delta^{-\frac{d}{\nu_1+d/2}}).
\end{align*}
Thus, in order to make $h_{G_n,V}\leq C_0m^{-\frac{2\nu_1}{d}}$, we set $\delta = C_0m^{-\frac{2\nu_1}{d}}$. Thus, as long as $n\geq C_{12}\exp(C_{13}m^{\frac{2\nu_1}{\nu_1+d/2}})$, $h_{G_n,V}\leq C_0m^{-\frac{2\nu_1}{d}}$ holds. This implies $m\leq C_{14}(\log n)^{\frac{\nu_1+d/2}{2\nu_1}}$. Since we require $\log m>1$, $n$ should satisfy $n>\exp((e/C_{14})^{\frac{2\nu_1}{\nu_1+d/2}})=: N_0$. Plugging $m\leq C_{14}(\log n)^{\frac{\nu_1+d/2}{2\nu_1}}$ in \eqref{eq:pflimsI1X}, \eqref{eq:coronlmI2}, and \eqref{eq:pflimspe2}, we finish the proof.
\end{proof}

\section{Sample path}\label{sec:samplepath}

\subsection{Linear kernel}\label{sec:samplepathlinear}

The focus of this section is to study how the unknown  parameters in the proposed linear kernel \eqref{eq:lk} affect the generated sample paths. We focus on the Mat\'ern kernel function, which has the form of \eqref{matern2}. 
Three types of unknown parameters are studied, including the $d$ positive diagonal elements of the diagonal matrix $\Theta$, denoted by $\boldsymbol{\theta}$, the positive scalar $\sigma^2$, and the smoothness parameter  $\nu$.

The sample paths are generated by the input functions  $g(x)=\sin(\alpha x)$, where $x\in \Omega=[0,2\pi]$ and $\alpha\in[0,1]$.
The value $\alpha$ indicates the frequency of the periodic function and the RKHS norm of $g$, i.e., $\|g\|_{\mathcal{N}_{\Psi}(\Omega)}$, increases monotonically  with respect to $\alpha$.
As a result, this input function creates an analogy to the sample paths in conventional GP by studying the paths as a function of $\alpha$ with different parameter settings.
In Figure \ref{fig:sample_path_linear}, the sample paths are demonstrated with different settings of the three types of parameters. The first row illustrates the sample paths with three different settings of $\nu$, given $\theta=1$ and $\sigma^2=1$. It appears that the smoothness of the resulting sample paths is not significantly affected by the setting of $\nu$, which typically controls the smoothness in conventional GPs.  
This is mainly because, unlike the conventional Mat\'ern kernel function \citep{stein2012interpolation}, the derivative of \eqref{eq:lk} with respect to the input $g$ is not directly related to the parameter $\nu$.
The middle panels in Figure \ref{fig:sample_path_linear} demonstrate the sample paths with different settings of $\theta$, given $\nu=2.5$ and $\sigma^2=1$. It shows that, as $\theta$ increases, the number of the local maxima and minima increases, which agrees with the observations in  conventional GPs. Lastly, the bottom three panels show the sample paths with different settings of $\sigma^2$, given $\nu=2.5$ and $\theta=1$. Similar to conventional GPs, $\sigma^2$ controls the amplitude of the resulting sample paths.

\subsection{Nonlinear kernel}\label{sec:samplepathnonlinear}

Based on the nonlinear kernel of \eqref{eq:nonlinearkernel} with the Mat\'ern kernel function defined in \eqref{matern1}, the sample paths are studied with respect to different settings of the $\gamma$, $\nu$ and $\sigma^2$.  As shown in Figure \ref{fig:sample_path_nonlinear}, the results appear to be consistent with the observations in conventional GPs where $\nu$ controls the smoothness of the function, $\theta$ controls the number of the local maxima and minima, and  $\sigma^2$ controls the amplitude of the functions.

\section{Supporting tables and figures in Sections 4 and 5}
The tables and figures that present the results in Sections 4 and 5 are provided in this section.

\begin{table}[h]
    \centering
    \begin{tabular}{C{1.2cm}|C{1.2cm}|C{1cm}|C{1cm}|C{1.1cm}|C{1.2cm}|C{1.5cm}|C{1.3cm}|C{2cm}}
     \toprule
     $g(\mathbf{x})$ & $x_1+x_2$ & $x_1^2$ &$x_2^2$&$1+x_1$&$1+x_2$&$1+x_1x_2$&$\sin(x_1)$&$\cos(x_1+x_2)$\\
     \midrule
     $f_1(g)$  & 1 & 0.33 & 0.33& 1.5& 1.5& 1.25& 0.46&0.50\\
     $f_2(g)$  & 1.5 & 0.14 & 0.14& 3.75& 3.75& 2.15& 0.18&0.26\\  
     $f_3(g)$  & 0.62 & 0.19 & 0.19& 0.49& 0.49& 0.84& 0.26&0.33\\     
    \bottomrule
    \end{tabular}
    \caption{Training data set for the numerical study, where $f_1(g)=\int_\Omega\int_\Omega g(\mathbf{x}){\rm{d}}x_1{\rm{d}}x_2$, $f_2(g)=\int_\Omega\int_\Omega g(\mathbf{x})^3{\rm{d}}x_1{\rm{d}}x_2$, and $f_3(g)=\int_\Omega\int_\Omega \sin(g(\mathbf{x})^2){\rm{d}}x_1{\rm{d}}x_2$.}
    \label{tab:linear_simulation}
\end{table}

\begin{table}[]
\begin{center}
\begin{tabular}{ c|c|c|c|c } 
 \toprule
Measurements & Method  & $f_1(g)=\int_\Omega\int_\Omega g$ & $f_2(g)=\int_\Omega\int_\Omega g^2$& $f_3(g)=\int_\Omega\int_\Omega \sin(g)$ \\ 
 \midrule
 \multirow{3}{*}{MSE} &  \texttt{FIGP}& $\boldsymbol{6.4\times 10^{-10}}$ & \textbf{0.012} & \textbf{0.016}\\ 
 & \texttt{FPCA} & $1.8\times 10^{-4}$ & 0.124 & 0.023\\ 
  & \texttt{T3} & 0.093 & 1.271 & 0.047\\ 
 \midrule
\multirow{3}{*}{Coverage (\%)} &  \texttt{FIGP}& \textbf{96.33} & 100 & \textbf{100}\\ 
 & \texttt{FPCA}& 100 & \textbf{92.33} & 76.00\\ 
  & \texttt{T3} & 100 & 98.33 & \textbf{100}\\ 
  \midrule
\multirow{3}{*}{Score} &  \texttt{FIGP}& \textbf{14.899} & \textbf{2.571} & \textbf{3.458}\\ 
 & \texttt{FPCA}& 6.631 & 1.207 & 0.290\\ 
   & \texttt{T3} & 1.064 & -1.364 & 2.047\\ 
 \bottomrule
\end{tabular}
\end{center}
    \caption{Prediction results of the FIGP and basis-expansion approach for the synthetic examples (\texttt{FPCA} indicates an FPCA expansion approach and \texttt{T3} indicates the Taylor series expansion of degree 3), including MSEs, average coverage rates of the 95\% prediction intervals, and the average proper scores, in which the values with better performances are boldfaced.}
    \label{tab:prediction_comparison}
\end{table}

\newpage

\begin{figure}[h!]
    \centering
    \includegraphics[width=0.95\textwidth]{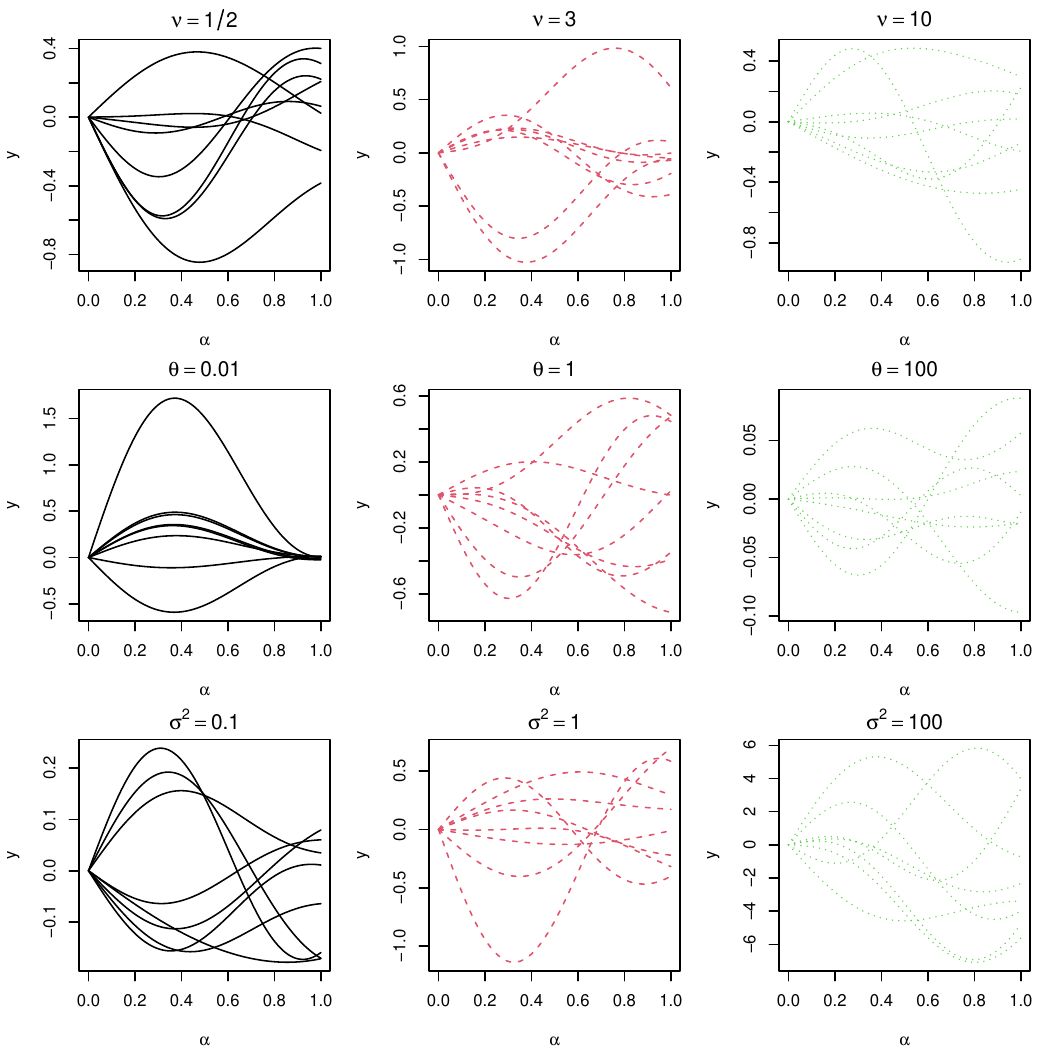}
    \caption{Sample paths of linear kernels. Top panel shows the effect of varying the parameter $\nu$ with the fixed $\theta=1$ and $\sigma^2=1$, middle panel shows the effect of varying the parameter $\theta$ with the fixed $\nu=2.5$ and $\sigma^2=1$, and the bottom panel shows the effect of varying the parameter $\sigma^2$ with the fixed $\nu=2.5$ and $\theta=1$.}
    \label{fig:sample_path_linear}
\end{figure}

\begin{figure}[h!]
    \centering
    \includegraphics[width=0.95\textwidth]{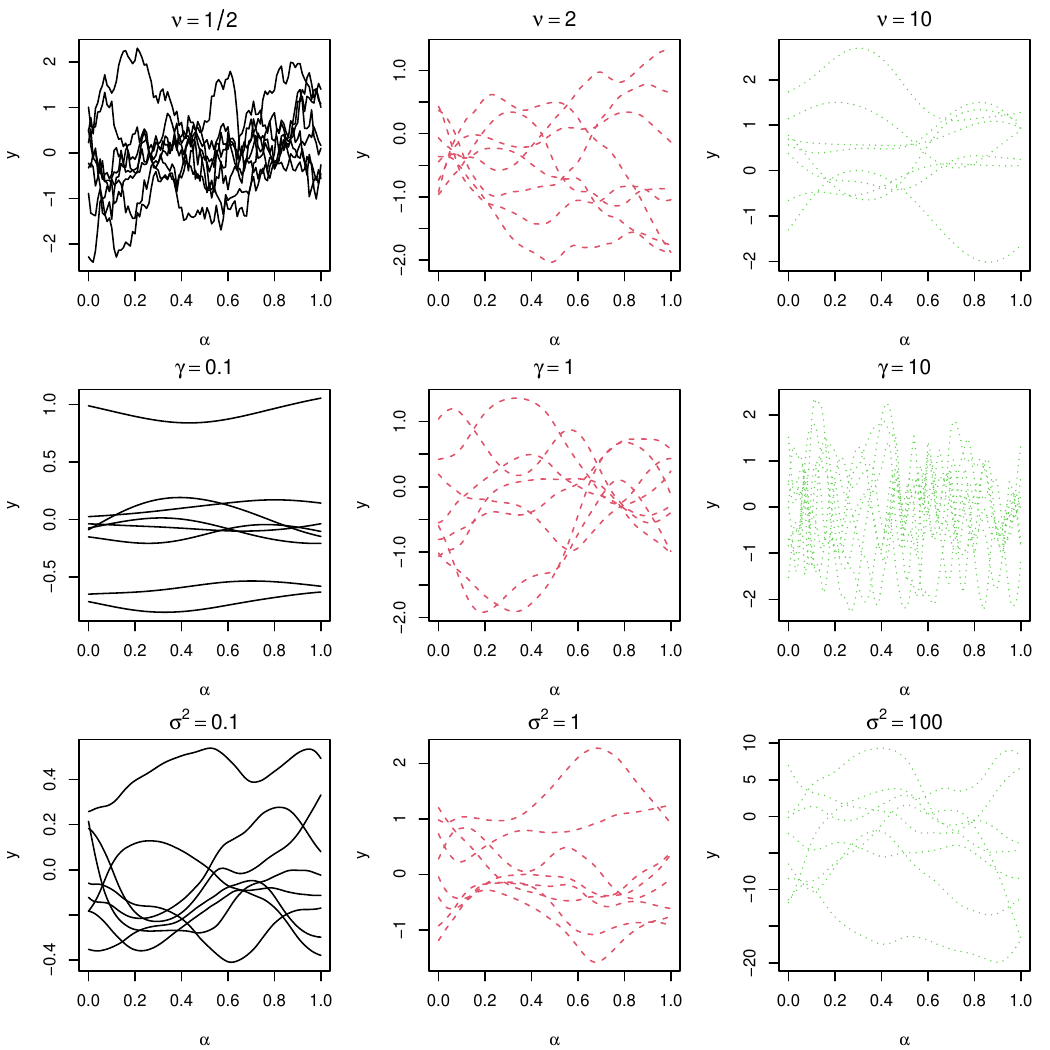}
    \caption{Sample paths of nonlinear kernels. Top panel shows the effect of varying the parameter $\nu$ with the fixed $\gamma=1$ and $\sigma^2=1$, middle panel shows the effect of varying the parameter $\gamma$ with the fixed $\nu=2.5$ and $\sigma^2=1$, and the bottom panel shows the effect of varying the parameter $\sigma^2$ with the fixed $\nu=2.5$ and $\gamma=1$.}
    \label{fig:sample_path_nonlinear}
\end{figure}

\begin{figure}[]
    \centering
    \includegraphics[width=0.7\textwidth]{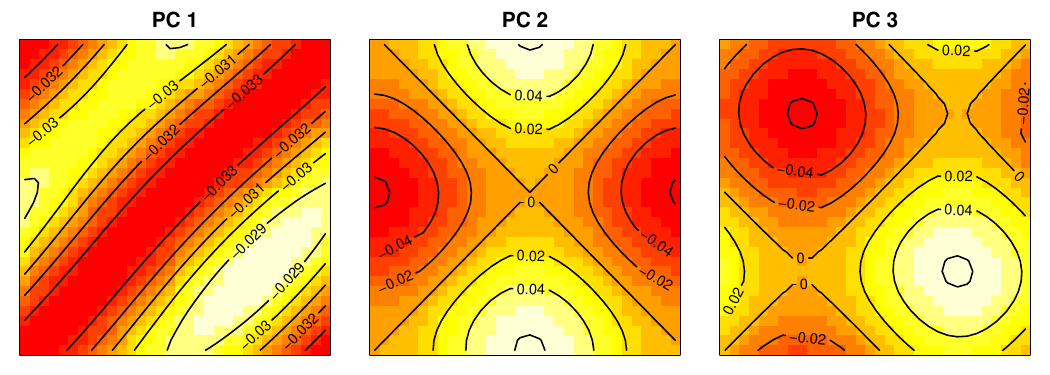}
    \caption{Principle components, which explain more than 99.99\% variations of the data.}
    \label{fig:realcase_pc}
\end{figure}

\begin{figure}[]
    \centering
    \includegraphics[width=0.8\textwidth]{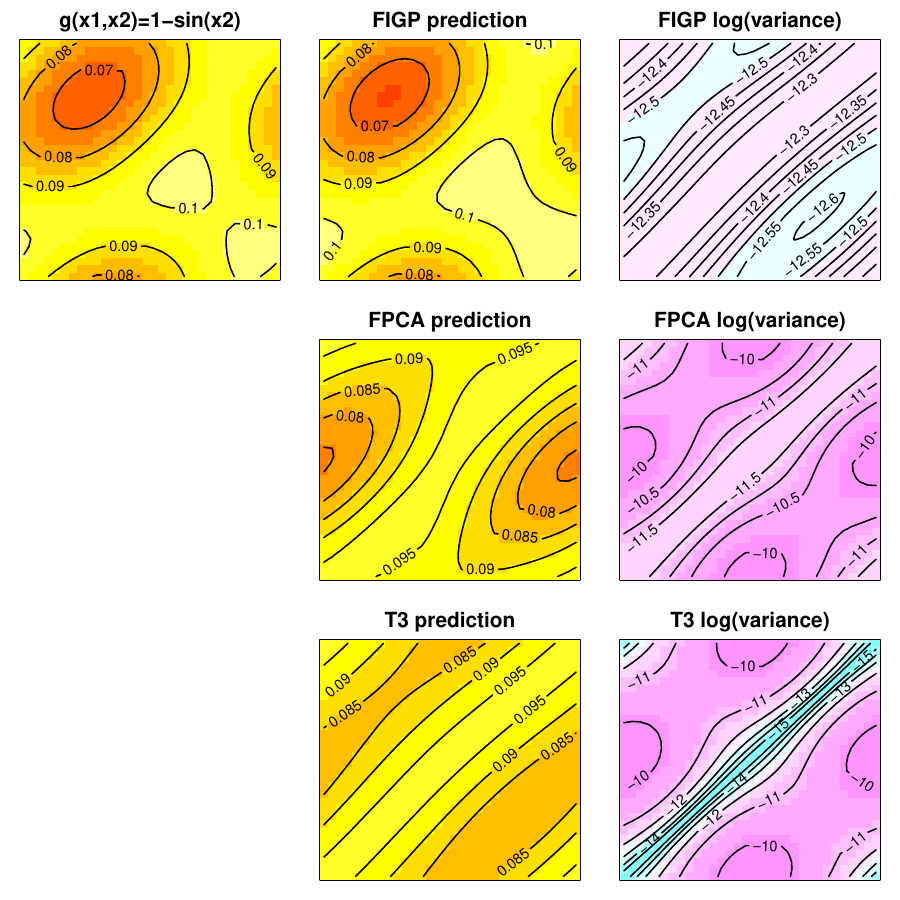}
    \caption{Prediction on the validation function. The left panel is the true output of the functional input  $g(x_1,x_2)=1-\sin(x_2)$, and the middle panels are the predictions of \texttt{FIGP}, \texttt{FPCA}, and \texttt{T3}, and the right panels are their variances in logarithm.}
    \label{fig:realcase_prediction}
\end{figure}

\end{document}